\documentclass[journal,twoside,web]{IEEEtran}

\usepackage{textcomp}
\usepackage{graphicx}
\usepackage{xcolor}

\usepackage{bbm}
\usepackage{mathrsfs}
\usepackage{verbatim}
\usepackage{amsmath}
\usepackage{amsfonts}
\usepackage{bm}
\usepackage{caption}
\usepackage{float}

\usepackage{amsthm}
\usepackage{mathtools}
\usepackage{cite}
\usepackage{hyperref}
\let\inf\relax 
\DeclareMathOperator*\inf{\vphantom{p}inf}
\usepackage{subfigure}
\allowdisplaybreaks

\theoremstyle{definition}
\newtheorem{assumption}{Assumption}
\newtheorem{definition}{Definition}
\newtheorem{remark}{Remark}
\newtheorem{problem}{Problem}

\theoremstyle{plain}

\newtheorem{theorem}{Theorem}
\newtheorem{lemma}{Lemma}

\begin{document}

\title{Approximate Information States for Worst-Case Control and Learning in Uncertain Systems}

\author{Aditya Dave, {\IEEEmembership{Member, IEEE,}} Nishanth Venkatesh, {\IEEEmembership{Student Member, IEEE,}} and Andreas A. Malikopoulos, {\IEEEmembership{Senior Member, IEEE}} 
	\thanks{This research was supported by NSF under Grants CNS-2149520 and CMMI-2219761.}
	\thanks{The authors are with the School of Civil and Environmental Engineering, Cornell University, Ithaca, NY (email: \texttt{a.dave@cornell.edu; ns924@cornell.edu; amaliko@cornell.edu).}} }

\maketitle

\begin{abstract}
In this paper, we investigate discrete-time decision-making problems in uncertain systems with partially observed states. We consider a non-stochastic model, where uncontrolled disturbances acting on the system take values in bounded sets with unknown distributions. 
We present a general framework for decision-making in such problems by using the notion of the information state and approximate information state, \textcolor{black}{and introduce conditions to identify an uncertain variable that can be used to compute an optimal strategy through a dynamic program (DP). 
Next, we relax these conditions and define \textit{approximate information states} that 
can be learned from output data without knowledge of system dynamics. We use approximate information states to formulate a DP that yields a strategy with a bounded performance loss.} 
Finally, we illustrate the application of our results in control and reinforcement learning using numerical examples.
\end{abstract}


\IEEEpeerreviewmaketitle

\section{Introduction}
\label{section:Introduction}

Decision-making under incomplete information is a fundamental problem in modern engineering applications involving cyber-physical systems \cite{kim2012cyber}, e.g., connected and automated vehicles \cite{Malikopoulos2020} and
social media platforms \cite{Dave2020SocialMedia}. In such applications, an agent must sequentially select control inputs to a dynamic system using only partial observations while accounting for uncontrolled disturbances that interfere with the system's evolution.
The most common modeling paradigm for such decision-making problems is \textit{the stochastic approach,} where all disturbances to the system are considered as random variables with known distributions, and the agent selects a decision-making strategy to minimize the \textit{expected incurred cost} \cite{varaiya_book}. Stochastic models have been utilized for problems in both 
control theory \cite{Malikopoulos2016b, ahmadi2020control, mahajan2012, dave2019decentralized, Dave2021a, Dave2021nestedaccess, Malikopoulos2021} and reinforcement learning \cite{ mishra2021decentralized, zhang2021multi,Malikopoulos2022a,kao2022common,Malikopoulos2024}. 
A strategy derived using the stochastic approach performs optimally on average 
across numerous operations of the system.
However, this performance degrades with any mismatch between the distributions assumed in modeling and the frequency of realizations encountered in implementation \cite{mannor2007bias}. Furthermore, many safety-critical applications require guarantees on the agent's performance during each operation \cite{brunke2022safe}, making expected cost an inadequate measure.

\textit{The non-stochastic approach} is an alternate modeling paradigm for safety-critical systems, where all disturbances are considered to belong to known sets with unknown distributions. The agent aims to select a decision-making strategy that minimizes the \textit{worst-case incurred cost} across a time horizon \cite{bacsar2008h}. Because this approach focuses on robustness against worst-case disturbances, the resulting strategy yields more conservative decisions than the stochastic approach. At the expense of average performance, we receive guarantees for each system operation. Thus, this approach has been applied to adversarial settings, e.g., cyber-security \cite{rasouli2018scalable} or cyber-physical systems \cite{shoukry2013minimax}, and failure-critical settings, e.g., water reservoirs \cite{giuliani2021state} or power systems \cite{zhu2011robust}.

In this paper, we propose a framework for non-stochastic decision-making using only partial observations in a dynamic system. When the system's dynamics are known to the agent, this problem falls under the purview of control theory \cite{aastrom2014control}. However, many applications involve decision-making with incomplete knowledge of the dynamics as, e.g., mixed traffic driving \cite{ venkatesh2023stochastic} and human-robot coordination \cite{hadfield2016cooperative, faros2023adherence}, or decision-making without a reliable state-space model, e.g., healthcare \cite{fatemi2021medical}. These restrictions yield a reinforcement learning problem \cite{kiumarsi2017optimal, recht2019tour}. 
To account for both cases, we formulate our problem using only output variables without assuming a known state-space model. In our exposition, we present rigorous definitions for the notions of \textit{information states} and \textit{approximate information states}. 
Using these notions, a surrogate state-space model can be constructed from output variables.   
This surrogate model can be used to formulate a control problem with full-state observation, whose solution yields either an optimal or an approximate strategy of the original problem. In reinforcement learning problems, the surrogate model can be learned from output data. Given a surrogate model, the agent can derive a decision-making strategy using standard techniques \cite{bertsekas2005dynamic}.  

\vspace{-6pt}

\subsection{Related Work}

\textit{1) Control theory:} There have been numerous research efforts in control theory to study dynamic decision-making problems given the system dynamics. For both stochastic and non-stochastic models, an agent can derive an optimal decision-making strategy offline using a dynamic programming (DP) decomposition \cite{bertsekas1971control}. 
For systems with perfectly observed states, the agent's optimal action at each instance of time is simply a function of the state. Using this property in a DP facilitates the efficient computation of an optimal control strategy \cite{hernandez2012discrete, moon2015minimaxcontrol}. In contrast, for partially observed systems, any optimal action is a function of the agent's entire memory of past observations and actions, which grows in size with time \cite{bertsekas2012dynamic}. Subsequently, the domain of the optimal strategy grows with time, and the DP for the problem requires a large number of computations \cite{papadimitriou1987complexity}.
This concern is alleviated using an \textit{information state} instead of the memory \cite{bernhard1996separation}. 

The most commonly used information state in stochastic control is the \textit{belief state}, i.e., a distribution on the state space conditioned on the agent's memory \cite{puterman2014markov, Dave2020a}. A general notion of information states for stochastic control was recently defined in \cite{subramanian2019approximate}. 
For non-stochastic control problems, the DP decomposition has been simplified using two well-known information states: (1) the \textit{conditional range}, which is the set of feasible states at any time consistent with the agent's memory \cite{cong2021rethinking} and can be used in both terminal cost \cite{ bertsekas1973sufficiently, gagrani2017decentralized, Dave2021minimax, moitie2002optimal} and instantaneous cost problems \cite{piccardi1993infinite, bernhard2003minimax, gagrani2020worst}; and (2) the \textit{maximum cost-to-come}, which is the maximum accrued cost functional \cite{james1994risk} used in additive cost problems \cite{bernhard2000max, coraluppi1999risk, Dave2023infhorizon}. 

A general notion of an information state for non-stochastic terminal cost problems was presented in \cite{Dave2022approx}. Information states have also been derived for risk-sensitive formulations \cite{baauerle2017partially} and distributionally robust formulations \cite{osogami2015robust}. \textcolor{black}{A related notion of symbolic abstractions has been developed in non-stochastic settings by drawing upon system approximation using bisimulation \cite{girard2012controller}. Typically, symbolic abstractions discretize the state and action spaces of continuous systems \cite{tazaki2008finite}. Such abstractions and their worst-case performance were studied for perfectly observed control problems in \cite{de2013symbolic, reissig2018symbolic}.}

The advantage of using information states (or other state abstractions) is that they result in a smaller approximate space. Thus, they generally yield a more efficient DP decomposition than the entire memory. However, in problems with large state spaces, utilizing information states may still not be practical \cite{gallestey1999max, saldi2014asymptotic}, creating a need for principled approximations. 

\textit{2) Reinforcement learning:} The literature on reinforcement learning is concerned with decision-making when the agent does not have prior knowledge of the system's dynamics \cite{meyn2022control}. For systems with perfectly observed states, these problems have been addressed using a variety of approaches \cite{levine2020offline}. 
In the stochastic formulation, both model-based \cite{malikopoulos2009real, moerland2020model} and model-free approaches \cite{ramirez2022model} have been utilized. 
In the non-stochastic formulation, the worst-case reinforcement learning problem was formulated and analyzed in \cite{morimoto2005robust}. Worst-case Q-learning was proposed for reinforcement learning problems in \cite{heger1994consideration,jiang1998minimax, garcia2015comprehensive}
and extended to problems with output feedback and known dynamics in \cite{valadbeigi2019h}. Actor-critic methods \cite{chakravorty2003minimax} and model-based off-policy learning approaches \cite{kiumarsi2017h} have also been developed for robust control. Alternate approaches using online adaptive algorithms were proposed in \cite{agarwal2019online, gradu2020non}. 
However, in general, reinforcement learning with partial observations remains challenging, especially for non-stochastic formulations. 

For stochastic formulations, the notion of \textit{approximate information states} was presented in \cite{subramanian2022approximate} to address the challenges of control and learning with partial observations. Approximate information states can construct a system approximation to improve the computational tractability of control problems at the cost of a bounded loss in performance. 
In reinforcement learning, approximate information states can be learned from output data to facilitate standard learning algorithms. Their performance has been empirically validated in robotics \cite{yang2022discrete} and healthcare \cite{killian2020empirical}.
No general theory of approximate information states exists yet for non-stochastic formulations.

\vspace{-6pt}

\subsection{Contributions and Organization}

In this paper, we develop a non-stochastic theory of approximate information states for both instantaneous and terminal cost problems, which can facilitate computationally efficient control and provide a principled approach to reinforcement learning using partial observations. 
The contributions of this paper are: (1) the introduction of a general notion of \textit{information states} (Definition \ref{def_info_state}) which yields an optimal DP decomposition for worst-case control (Theorem \ref{thm_opt_dp}); 
(2) the introduction of the notion of \textit{approximate information states} (Definition \ref{def_approx}) that can either be constructed from output variables or learned from output data (Subsection \ref{subsection:learning_algorithm}); (3) the formulation of an approximate DP (Theorem \ref{thm_approx_term_dp}) that computes a control strategy with a bounded loss of optimality (Theorem \ref{thm_approx_term_policy}); and (4) the exposition of examples of approximate information states (Subsection \ref{subsection:approx_examples}) with corresponding theoretical guarantees (Theorems \ref{thm_main_ap_a} - \ref{thm_main_ap_b}). We also illustrate our results 
using numerical examples (Subsection \ref{section:example}).

Note that while our theory shares conceptual similarities to the theory of approximate information states for stochastic problems in \cite{subramanian2022approximate}, \textcolor{black}{our focus on non-stochastic problems necessitates the use of a distinct formulation with uncertain variables \cite{nair2013nonstochastic} and a different mathematical approach. Our results bound the worst-case approximation loss rather than the expected loss.} We reported preliminary results for terminal-cost control problems in \cite{Dave2022approx}. \textcolor{black}{This paper extends the preliminary work as follows: (1) we consider worst-case instantaneous cost problems which subsume terminal cost problems; (2) we allow all variables to take values in continuous spaces; (3) we derive explicit bounds using state-discretization in both perfectly and partially observed systems;
and (4) we illustrate the application of our results to a reinforcement learning problem.}

The remainder of the paper proceeds as follows. In Section \ref{section:problem}, we present our problem formulation. In Section \ref{section:info_state}, we define the notion of information states and prove the optimality of the corresponding DP decomposition. In Section \ref{section:approx}, we present the notion of approximate information states, a resulting approximate DP, and theoretical bounds on the approximation loss. 
In Section \ref{section:example}, we present a numerical example to illustrate the application of our results. 
In Section \ref{section:conclusion}, we draw concluding remarks and discuss future work.

\vspace{-6pt}

\section{Modeling Framework}
\label{section:problem}

\subsection{Preliminaries}

\textbf{1) Uncertain Variables:}
In this paper, we utilize the mathematical framework for \textit{uncertain variables} from \cite{nair2013nonstochastic}. 
An uncertain variable is a non-stochastic analogue of a random variable with set-valued uncertainty. 
For a sample space $\Omega$ and a set $\mathcal{X}$, an uncertain variable is a mapping $X: \Omega \to \mathcal{X}$. For any $\omega \in \Omega$, it has the realization $X(\omega) = x \in \mathcal{X}$. 
The \textit{marginal range} of $X$ is the set $[[X]] \hspace{-1pt} := \hspace{-1pt} \{X(\omega) \; | \; \omega \in \Omega\}$. For two uncertain variables $X \in \mathcal{X}$ and $Y \in \mathcal{Y}$, their \textit{joint range} is $[[X,Y]] \hspace{-1pt} := \hspace{-1pt} \{ \big(X(\omega), Y(\omega) \big) \; | \; \omega \in \Omega \}$. For a given realization $y$ of $Y$, the \textit{conditional range} of $X$ is $[[X|y]] \hspace{-1pt} := \hspace{-1pt} \{ X(\omega) \; | \; Y(\omega) $ $= y, \; \omega \in \Omega \}$ \textcolor{black}{and, generally,
$[[X|Y]] \hspace{-1pt} := \hspace{-1pt} \cup_{y \in [[Y]]} [[X|y]]$.}

\textbf{2) Hausdorff Distance:} Consider that the $\mathcal{X}, \mathcal{Y}$ are nonempty subsets of a metric space $(\mathcal{S},\eta)$, where $\eta(\cdot, \cdot)$ is the metric. 
Then, the \textit{Hausdorff distance} between $\mathcal{X}$ and $\mathcal{Y}$ is
\begin{align} \label{H_met_def}
    \mathcal{H}(\mathcal{X}, \mathcal{Y}) := \max \Big\{ \sup_{x \in \mathcal{X}} \inf_{y \in \mathcal{Y}} \eta(x, y), \sup_{y \in \mathcal{Y}} \inf_{x \in \mathcal{X}} \eta(x, y) \Big\}.
\end{align}
When the two sets $\mathcal{X}, \mathcal{Y}$ are bounded, the Hausdorff distance in \eqref{H_met_def} constitutes a pseudo-metric, i.e., $\mathcal{H}(\mathcal{X}, \mathcal{Y}) = 0$ if and only if ${closure}(\mathcal{X}) = {closure}(\mathcal{Y})$ \cite[Appendix]{girard2007approximation}. When both $\mathcal{X}, \mathcal{Y}$ are compact, the Hausdorff distance is a metric, i.e., $\mathcal{H}(\mathcal{X}, \mathcal{Y}) = 0$ if and only if $\mathcal{X} = \mathcal{Y}$ \cite[Chapter 1.12]{barnsley2006superfractals}. 

\textbf{3) L-invertible Functions:} Consider a function $f: \mathcal{X} \to \mathcal{Y}$. For any $y \in \mathcal{Y}$, the pre-image of the function is $f^{-1}(y) = \big\{ x \in \mathcal{X} ~|~ f(x) = y\big\}$. Then, the function $f$ 
is called \emph{$L$-invertible} if there exists a constant $L_{f^{-1}} \hspace{-2pt} \in \hspace{-2pt} \mathbb{R}_{\geq0}$ such that
\begin{gather} \label{L_inv_def}
    \mathcal{H}\big(f^{-1}(y^1), f^{-1}(y^2)\big) \hspace{-2pt} \leq \hspace{-2pt} L_{f^{-1}} {\cdot} \eta(y^1,y^2), \; \; \forall y^1, y^2 \in \mathcal{Y}.
\end{gather}
For uncertain variables $X \in \mathcal{X}$ and $Y \in \mathcal{Y}$ such that $Y=f(X)$, the pre-image of $f$ given a realization $y \in [[Y]]$ is $[[X|y]]$, i.e., $f^{-1}(y) = [[X|y]]$. Thus, if $f$ is \textit{$L$-invertible}, we equivalently state that for all $y^1, y^2 \in [[Y]]$ and $L_{X|Y} = L_{f^{-1}}$:
\begin{gather} \label{L_inv_2}
    \mathcal{H}\big([[X|y^1]], [[X|y^2]]\big) \leq L_{X|Y} \cdot \eta(y^1,y^2).
\end{gather}

\textcolor{black}{\textbf{4) Notation:} Uncertain variables will be denoted by uppercase alphabets and their realizations by lower-case alphabets. Variables and functions with a bar on them, such as $\bar{V}(\cdot)$, pertain to information states, and those with a hat on them, such as $\hat{V}(\cdot)$, pertain to approximate information states.}

\vspace{-6pt}

\subsection{Problem Formulation}

We consider an agent that seeks to control an uncertain system over $T \in \mathbb{N}$ discrete time steps. At each time $t=0,\dots,T$, the agent receives an observation from the system, denoted by the uncertain variable $Y_t \in \mathcal{Y}_t$, and generates a control action denoted by the uncertain variable $U_t \in \mathcal{U}_t$. After generating the action at each $t$, the agent incurs a cost denoted by the uncertain variable $C_t \in \mathcal{C}_t \subset \mathbb{R}_{\geq0}$. 
To account for a possibly unknown state-space model, we describe the dynamics using an \textit{input-output model}, as follows. At each $t$, the system receives two inputs: the control action $U_t$, and an uncontrolled disturbance $W_t \in \mathcal{W}_t$. 
\textcolor{black}{The uncontrolled disturbances $\{W_t\,|\,t=0,\dots,T\}$ constitute a sequence of independent uncertain variables that enter the system through a channel separate from the control action.}
After receiving the inputs at each $t$, the system generates two outputs:
\begin{align}
    Y_{t+1} &= h_{t+1}(W_{0:t}, U_{0:t}), \\
    C_t &= d_t(W_{0:t}, U_{0:t}),
\end{align}
using some observation function $h_{t+1}$ 
and cost function $d_t$. 
The initial observation is generated as $Y_0 = h_0(W_0)$.

The agent has a perfect recall of the history of observations and control actions. The memory of the agent at each $t$ is denoted by the uncertain variable $M_t := (Y_{0:t}, U_{0:t-1})$, which takes values in the set $\mathcal{M}_t := \prod_{\ell=0}^t \mathcal{Y}_{\ell} \times \prod_{\ell=0}^{t-1} \mathcal{U}_{\ell}$. The agent uses the memory $M_t$ and a control law $g_t: \mathcal{M}_t \to \mathcal{U}_t$ at each $t$ to generate the action $U_t = g_t(M_t)$. We denote the control strategy by $\boldsymbol{g} := (g_0,\dots,g_T)$ and the set of all feasible control strategies by $\mathcal{G}$. The performance of a strategy $\boldsymbol{g} \in \mathcal{G}$ is measured by the \textit{worst-case or maximum instantaneous cost}
    \begin{align} \label{eq_instantaneous_criterion}
    \mathcal{J}(\boldsymbol{g}) := \max_{t=0,\dots,T} \sup_{w_{0:t} \in [[W_{0:t}]]} C_t.
    \end{align}

\begin{problem} \label{problem_1}
\textcolor{black}{The agent's robust control problem is to derive a control strategy $\boldsymbol{g}^* \in \mathcal{G}$, such that $\mathcal{J}(\boldsymbol{g}^*) \leq \mathcal{J}(\boldsymbol{g})$ for all $\boldsymbol{g} \in \mathcal{G}$,
given the marginal ranges $\{[[U_t]],[[W_t]],[[C_t]],[[Y_t]] ~|~ t=0,\dots,T\}$ and the functions $\{h_t, d_t ~|~ t=0,\dots,T\}$.}
\end{problem}

We seek to tractably compute an optimal strategy if one exists. In our framework, we impose the following assumptions:

\begin{assumption} \label{assumption_1}
We consider that the sets $\{\mathcal{U}_t,\mathcal{W}_t,\mathcal{Y}_t~|~t=0,\dots,T\}$ and $\{\mathcal{C}_t~|~t=0,\dots,T\}$ are all bounded subsets of a metric space $(\mathcal{S}, \eta)$ and $\mathbb{R}_{\geq0}$, respectively.
\end{assumption}

Assumption \ref{assumption_1} allows for both continuous and finite valued feasible sets while ensuring that the problem is well-posed. 

\begin{assumption} \label{assumption_2}
The observation functions $\{h_t~|~ t=0,\dots,T\}$ of the system are both Lipschitz and $L$-invertible, whereas the cost functions $\{d_t ~|~ t=0,\dots,T\}$ are Lipschitz continuous.
\end{assumption}

Assumption \ref{assumption_2} is satisfied by a large class of observation functions, including (1) all functions with compact domains and finite co-domains and (2) bi-Lipschitz functions 
with compact domains and compact co-domains (see Appendix A). We require both assumptions 
our main results.

\begin{remark}
In our exposition, we also consider a special case of \eqref{eq_instantaneous_criterion}, called the \textit{maximum terminal cost criterion}, given by
\begin{align} \label{eq_terminal_criterion}
    \mathcal{J}^\text{tm}(\boldsymbol{g}) := \sup_{w_{0:T} \in [[W_{0:T}]]} C_T.
\end{align}
In addition to results for Problem \ref{problem_1}, we often present specializations to terminal cost problems that are easier to interpret. 
\end{remark}

\begin{remark}
\textcolor{black}{Problem \ref{problem_1} is a robust control problem with known dynamics. Our results are derived primarily within this setting. However, our main results in Section \ref{section:approx} can also be used in learning problems where $h_t$ and $d_t$ are unknown for all $t$. We present a learning approach in Subsections \ref{subsection:learning_algorithm} and \ref{subsection:pursuit}.}
\end{remark}


\vspace{-6pt}

\section{Dynamic Programs and Information States} \label{section:info_state}

\textcolor{black}{In this section, we first present a memory-based DP to compute the optimal solution to Problem \ref{problem_1}.
This will serve as a reference for analyzing subsequent DPs. 
Then, we highlight the DP's computational challenges and introduce information states in Subsections \ref{subsection:info_state} and \ref{subsection:alt_info} to alleviate them. Subsection \ref{subsection:info_examples} presents examples of information states.}

 To arrive at the memory-based DP, we construct a ``new'' perfectly observed system whose state at each $t$ is the memory $M_t$, evolving as $M_{t+1} = (M_t, U_t, Y_{t+1})$. 
 For realizations $m_t \hspace{-1pt} \in \hspace{-1pt} [[M_t]]$ and $u_t \hspace{-1pt} \in \hspace{-1pt} [[U_t]]$, the maximum cost at time $t$ is
\begin{align*}
    \sup_{w_{0:t} \in [[W_{0:t}]]} \hspace{-4pt} C_t = \hspace{-2pt} \sup_{c_t \in [[C_t]]^{\boldsymbol{g}}} \hspace{-4pt} c_t = \hspace{-4pt} \sup_{m_t, u_t \in [[M_t, U_t]]^{\boldsymbol{g}}} \sup_{c_t \in [[C_t|m_t, u_t]]^{\boldsymbol{g}}} \hspace{-4pt} c_t, 
\end{align*}
for all $t=0,\dots,T$, where $[[C_t]]^{\boldsymbol{g}}$, $[[M_t, U_t]]^{\boldsymbol{g}}$ and $[[C_t|m_t, u_t]]^{\boldsymbol{g}}$ are the respective marginal ranges and the conditional range induced by strategy $\boldsymbol{g}$. \textcolor{black}{Recall that $m_t = (y_{0:t},$ $u_{0:t-1})$ and thus, we can expand the conditional range as
\begin{align} \label{new_state_independence_g}
    &[[C_t|m_t, u_t]]^{\boldsymbol{g}}  = \big\{ c_t \in \mathcal{C}_t ~\big|~ \exists \; w_{0:t} \in [[W_{0:t}]] \text{ such that } \nonumber \\
    &c_t = d_t(w_{0:t}, u_{0:t}), \; y_0 = h_0(w_0), \;y_\ell = h_\ell(w_{0:\ell}, u_{0:\ell-1}), \nonumber \\
    &\forall \ell = 1,\dots,t \big\} 
    = [[C_t|m_t,u_t]],
\end{align}
where, the set of feasible disturbances $[[W_{0:t}]]$, cost function $d_t(\cdot)$, and observation functions $h_{\ell}(\cdot)$ for all $\ell = 0,\dots,t$ are independent of the choice of strategy $\boldsymbol{g}$. Thus, $[[C_t|m_t, u_t]]^{\boldsymbol{g}}$ is independent of the choice of strategy $\boldsymbol{g}$ and we can drop $\boldsymbol{g}$.} Next, we define $e_t(m_t, u_t) : = \sup_{c_t \in [[C_t|m_t, u_t]]} c_t$, independent of $\boldsymbol{g}$, and 
state that
\begin{align}
    \hspace{-2pt} \sup_{m_t, u_t \in [[M_t, U_t]]^{\boldsymbol{g}}}  \sup_{c_t \in [[C_t|m_t, u_t]]^{\boldsymbol{g}}} \hspace{-4pt} c_t &= \hspace{-4pt} \sup_{m_t, u_t \in [[M_t, U_t]]^{\boldsymbol{g}}} e_t(m_t, u_t) \nonumber \\
    &= \hspace{-4pt} \sup_{w_{0:t} \in [[W_{0:t}]]} \hspace{-4pt} e_t(M_t, U_t). \label{new_state_2}
\end{align}
Since $e_t(M_t, U_t)$ is a function of the new state $M_t$ and control action $U_t$, it serves as an incurred cost at each $t=0,\dots,T$ in our new perfectly observed system \cite{bernhard2003minimax}. The new instantaneous performance criterion is $\mathcal{E}(\boldsymbol{g}) := \max_{t=0,\dots,T} \sup_{w_{0:t} \in [[W_{0:t}]]}e_t(M_t, $ $U_t)$ and from \eqref{new_state_2}, 
$\mathcal{E}(\boldsymbol{g}) = \mathcal{J}(\boldsymbol{g})$ for any $\boldsymbol{g}$. 
Subsequently, any optimal strategy in the new system is optimal for Problem \ref{problem_1}. If such an optimal strategy exists, we can compute it using a standard DP for perfectly observed systems, as follows.
For all $t$, for each $m_t \in [[M_t]]$ and $u_t \in [[U_t]]$ we  recursively define the \textit{value functions}
\begin{align}
    Q_t(m_t, u_t) := &\max \Big \{ \sup_{c_t \in [[C_t|m_t, u_t]]} c_t,  \nonumber \\
    & \; \; \sup_{m_{t+1} \in [[M_{t+1}|m_t, u_t]]} V_{t+1}(m_{t+1}) \Big \}, \label{DP_basic_1}\\
    V_t(m_t) := &\inf_{u_t \in [[U_t]]} Q_t(m_t, u_t), \label{DP_basic_2}
\end{align}
where $V_{T+1}(m_{T+1}) := 0$, identically. We define the extra value function $V_{T+1}$ to ensure that the right-hand side (RHS) of \eqref{DP_basic_1} is well defined at time $T$. Then, we can show using standard arguments \cite{bertsekas1973sufficiently, bernhard2003minimax} that the optimal value of Problem \ref{problem_1} is $\inf_{\boldsymbol{g} \in \mathcal{G}}\mathcal{J}(\boldsymbol{g}) = \sup_{m_0 \in [[M_0]]}V_0(m_0)$. Furthermore, if there exists an action $u_t^* \in [[U_t]]$ which achieves the infimum in the RHS of \eqref{DP_basic_2}, then $g_t^*(m_t) := \arg \min_{u_t \in [[U_t]]}$ $Q_t(m_t, u_t)$ gives an optimal control law at each $t=0,\dots,T$ and the control strategy $\boldsymbol{g}^* = g_{0:T}^*$ is optimal for Problem \ref{problem_1}.

\begin{remark}
The DP \eqref{DP_basic_1} - \eqref{DP_basic_2} can be specialized to the terminal cost criterion \eqref{eq_terminal_criterion} by defining for all $t=0,\dots,T-1$,
\begin{align}
    Q_t^\text{tm}(m_t, u_t) := &\sup_{m_{t+1} \in [[M_{t+1}|m_t, u_t]]} V_{t+1}^\text{tm}(m_{t+1}), \label{DP_ad_term_1}\\
    V_t^\text{tm}(m_t) := &\inf_{u_t \in [[U_t]]} Q_t^\text{tm}(m_t, u_t), \label{DP_ad_term_2}
\end{align}
where $Q_T^\text{tm}(m_{T},u_T) := \sup_{c_T \in [[C_T|m_T, u_T]]} c_T$ and $V_{T}^\text{tm}(m_{T}) := \inf_{u_T \in  [[U_t]]} Q_T^\text{tm}(m_{T},u_T)$.
We will use this terminal cost DP to simplify the exposition in Section \ref{section:approx}.
\end{remark}

\begin{remark}
A valid argument referring to the minimum of the RHS of \eqref{DP_basic_2} at each $t =0,\dots,T$ is both a necessary and sufficient condition to ensure the existence of an optimal control strategy in Problem \ref{problem_1} \cite{bertsekas1973sufficiently, didinsky1995design}. 
Consider that marginal ranges of all uncertain variables are \textit{compact} rather than just bounded. From Assumption \ref{assumption_2}, the observation and cost functions at each $t$ are Lipschitz. Using these properties in \eqref{DP_ad_term_1} - \eqref{DP_ad_term_2}, we can show that the value functions are continuous and the conditional ranges are compact for all $t$, which implies that the minimum is achieved in the RHS of \eqref{DP_basic_2}. 
However, for generality, we continue using $\sup$ and $\inf$ in our exposition. 
\end{remark}

\begin{remark}
\textcolor{black}{In the RHS of \eqref{DP_basic_2} at each $t$, we are required to solve an optimization for each  $m_t \in [[M_t]]$. This is computationally challenging for longer horizons as the size of the set $[[M_t]]$ increases with time $t$. 
This concern necessitates an alternate DP decomposition that can derive an optimal control strategy while potentially achieving more favorable computational properties. We present such a DP decomposition in Subsection \ref{subsection:info_state} by identifying an uncertain variable, called an \textit{information state}, which can be used to generate an optimal control action at each time step instead of the memory.}
\end{remark}

\vspace{-12pt}

\subsection{Information States} \label{subsection:info_state}

In this subsection, we define the notion of information states, use them in a DP decomposition, and prove it is optimal. 
\begin{definition} \label{def_info_state}
An \textit{information state} 
at each $t=0,\dots,T$ is an uncertain variable $\Pi_t= \sigma_t(M_t)$ taking values in a bounded set $\mathcal{P}_t$, generated by a function $\sigma_t: \mathcal{M}_t \to \mathcal{P}_t$, and for all $m_t \in [[M_t]]$ and $u_t \in [[U_t]]$, satisfying the following properties:

1) \textcolor{black}{\textit{Sufficient to evaluate cost:} An information state can replace the memory to compute the worst-case cost, i.e.,}
\begin{gather}
    \sup_{c_t \in [[C_t|m_t, u_t]]} c_t = \sup_{c_t \in [[C_t|\sigma_t(m_t), u_t]]} c_t. \label{p1}
\end{gather}

2) \textcolor{black}{\textit{Sufficient to predict its own evolution:} An information state can replace the memory to compute its own conditional range at the next time step, independent of the strategy $\boldsymbol{g}$, i.e.,}
\begin{gather}
    [[\Pi_{t+1}|m_t, u_t]] = [[\Pi_{t+1}|\sigma_t(m_t), u_t]]. \label{p2}
\end{gather}
\end{definition}

We can use the information states from Definition \ref{def_info_state} directly in a DP, as follows. For all $t=0,\dots,T$, for all  $\pi_t \in [[\Pi_t]]$ and $u_t \in [[U_t]]$, we recursively define the value functions 
\begin{align}
    \bar{Q}_t(\pi_t, u_t) := &\max \Big \{\sup_{c_t \in [[C_t|\pi_t, u_t]]} c_t,  \nonumber \\
    & \; \; \; \sup_{\pi_{t+1} \in [[\Pi_{t+1}|\pi_t, u_t]]} \bar{V}_{t+1}(\pi_{t+1}) \Big \}, \label{DP_info_1}\\
    \bar{V}_t(\pi_t) := &\inf_{u_t \in [[U_t]]} \bar{Q}_t(\pi_t, u_t), \label{DP_info_2}
\end{align}
where $\bar{V}_{T+1}(\pi_{T+1}) := 0$ identically. If the minimum in the RHS of \eqref{DP_info_2} exists, this DP yields a control law as $\bar{g}_t^*(\pi_t) := \arg \min_{u_t \in [[U_t]]} \bar{Q}_t(\pi_t, u_t)$ at each $t $. Next, we prove that this DP computes the same value as the optimal DP \eqref{DP_basic_1} - \eqref{DP_basic_2}. 

\begin{theorem} \label{thm_opt_dp}
Let $\Pi_t = \sigma_t(M_t)$ be an information state at any $t$. Then, for all $t$, and for all $m_t \in [[M_t]]$ and $u_t \in [[U_t]]$, 
\begin{gather}
    \hspace{-5pt} {Q}_t(m_t, u_t) \hspace{-2pt} = \hspace{-2pt} \bar{Q}_t\big(\sigma_t(m_t), u_t\big) \text{ and } V_t(m_t) \hspace{-2pt} = \hspace{-2pt} \bar{V}_t\big(\sigma_t(m_t)\big). \label{eq_thm_opt_dp}
\end{gather}
\end{theorem}

\begin{proof}
Let $m_t \in [[M_t]]$ and $u_t \in [[U_t]]$ be given realizations of $M_t$ and $U_t$, respectively, for all $t=0,\dots,T$. We prove the result by mathematical induction starting at the last time step. At time $T+1$, \eqref{eq_thm_opt_dp} holds trivially because $V_{T+1}(m_{T+1}) = \bar{V}_{T+1}(\sigma_{T+1}(m_{T+1})) = 0$. This forms the basis of our induction. Next, for any $t=0,\dots,T$, we consider the induction hypothesis that $V_{t+1}(m_{t+1}) = \bar{V}_{t+1}(\sigma_{t+1}(m_{t+1}))$.
Given the hypothesis, we first prove that ${Q}_t(m_t, u_t) = \bar{Q}_t(\sigma_t(m_t), u_t)$ by comparing the RHS of \eqref{DP_basic_1} to the RHS of \eqref{DP_info_1} term by term. The first terms are equal by direct application of \eqref{p1} from Definition \ref{def_info_state}. Next, we use the induction hypothesis for the second term in the RHS of \eqref{DP_basic_1}, to state that
$\sup_{ m_{t+1} \in [[M_{t+1}|m_t, u_t]]}V_{t+1}(m_{t+1}) = \sup_{m_{t+1} \in [[M_{t+1}|m_t, u_t]]}\bar{V}_{t+1}(\sigma_{t+1}(m_{t+1}))$ 
$= \sup_{\sigma_{t+1}(m_{t+1}) \in [[\Pi_{t+1}|\sigma_t(m_t), u_t]]} \bar{V}_{t+1}(\sigma_{t+1}(m_{t+1}))$,
where, in the second equality, we use the fact that $[[\Pi_{t+1}|m_t, u_t]] = \big\{\sigma_{t+1}(m_{t+1}) \in \mathcal{P}_{t+1}\big| m_{t+1} \in [[M_{t+1}|m_t,u_t]]\big\}$ and \eqref{p2} from Definition \ref{def_info_state}. This establishes that the second term in the RHS of \eqref{DP_basic_1} equals the second term in the RHS of \eqref{DP_info_1} and given the induction hypothesis for time $t+1$, we have $Q_t(m_t, u_t) = \bar{Q}_t(\sigma_t(m_t), u_t)$. Next, we minimize both sides of the equality with respect to $u_t \in [[U_t]]$, and use the definitions in \eqref{DP_basic_2} and \eqref{DP_info_2} to write that 
$V_t(m_t) = \inf_{u_t \in \mathcal{U}_t} Q_t(m_t, u_t)$ $= \inf_{u_t \in \mathcal{U}_t}\bar{Q}_t\big(\sigma_t(m_t), u_t\big)$ $= V_t\big(\sigma_t(m_t)\big)$,
which proves the hypothesis at time $t$. The result follows for all $t$ using mathematical induction.
\end{proof}

Theorem \ref{thm_opt_dp} implies that \eqref{DP_info_1} - \eqref{DP_info_2} is an optimal DP decomposition for Problem \ref{problem_1}, i.e., if an optimal strategy exists for this DP, it is an optimal solution to Problem \ref{problem_1} as follows. Consider a control strategy $\bar{\boldsymbol{g}}^* = \bar{g}_{0:T}^*$ computed using \eqref{DP_info_1} - \eqref{DP_info_2}. A corresponding memory-based strategy $\boldsymbol{g} = g_{0:T}$ contains $g_t(m_t) := \bar{g}^*_t(\sigma_t(m_t))$ for all $m_t \in [[M_t]]$ and $t=0,\dots,T$. Then, from Theorem \ref{thm_opt_dp}, $\boldsymbol{g}$ achieves the infimum value at each $t$ to form an optimal solution to Problem \ref{problem_1}. 

\begin{remark}
In practice, using an information state to construct the DP decomposition is useful computationally only if, for most time steps in $t=0,\dots,T$, either the value functions in \eqref{DP_info_1} - \eqref{DP_info_2} have useful properties like concavity, or the set $\mathcal{P}_t$ is \textit{smaller} than $\mathcal{M}_t$ for some measure of size, such as the cardinality, diameter, and dimension. We present some examples of information states for different systems in Subsection \ref{subsection:info_examples}.
\end{remark}


\vspace{-12pt}

\subsection{Alternate Characterization of Information States} \label{subsection:alt_info}

When exploring whether an uncertain variable is a valid candidate for an information state, it may be difficult to verify \eqref{p2} in Definition \ref{def_info_state}. In this subsection, we present two \textit{stronger} conditions to replace \eqref{p2}. Specifically, at each $t = 0,\dots,T$, to establish that $\Pi_t = \sigma_t(M_t)$ is a valid information state, it is sufficient to satisfy the following conditions instead of \eqref{p2}:

\textit{1) State-like evolution:} There exists a function $\bar{f}_t : \mathcal{P}_t \times \mathcal{U}_t \times \mathcal{Y}_{t+1} \to \mathcal{P}_{t+1}$, independent of the strategy $\boldsymbol{g}$, such that
\begin{gather}
    \Pi_{t+1} = \bar{f}_t(\Pi_t, U_t, Y_{t+1}). \label{p2a}
\end{gather}

\textit{2) Predicts observations:} For all $m_t \in \mathcal{M}_t$ and $u_t \in \mathcal{U}_t$, 
\begin{gather}
    [[Y_{t+1}|m_t, u_t]] = [[Y_{t+1}|\sigma_t(m_t), u_t]], \label{p2b}
\end{gather}
where, both conditional ranges in \eqref{p2b} can be evaluated independently of the choice of strategy $\boldsymbol{g}$.

Next, we prove that these two conditions, in addition to \eqref{p1}, are sufficient to identify an information state.

\begin{lemma}
For all $t=0,\dots,T$, if an uncertain variable $\Pi_t = \sigma_t(M_t)$ satisfies \eqref{p2a} - \eqref{p2b}, it also satisfies \eqref{p2}.
\end{lemma}

\begin{proof}
For all $t=0,\dots,T$ and $m_t \in \mathcal{M}_t$, suppose that $\pi_t = \sigma_t(m_t)$ satisfy \eqref{p2a} - \eqref{p2b}. Then, we substitute \eqref{p2a} into the left hand side (LHS) of \eqref{p2} to state that
\begin{multline}
    [[\Pi_{t+1}|m_t,u_t]] = [[\bar{f}_t(\sigma_t(m_t), u_t, Y_{t+1})~|~m_t, u_t]] \\
    = \big\{\bar{f}_t(\sigma_t(m_t), u_t,y_{t+1}) ~\big|~ y_{t+1} \in [[Y_{t+1}|m_t,u_t]]\big\}, \label{prf_lem_1_1}
\end{multline}
where, in the second equality, we write the conditional range as a set. Next, using \eqref{p2b} on the range of observations in the conditioning of \eqref{prf_lem_1_1}, we can state that $\big\{\bar{f}_t(\sigma_t(m_t), u_t,y_{t+1}) ~\big|~ y_{t+1} \in [[Y_{t+1}|m_t,u_t]]\big\}=$ $\big\{\bar{f}_t(\sigma_t(m_t),u_t,y_{t+1}) ~\big|~ y_{t+1} \in [[Y_{t+1}|\sigma_t(m_t),u_t]]\big\}$ $= [[\bar{f}_t(\sigma_t(m_t), u_t, Y_{t+1})~|~\sigma_t(m_t), u_t]] \hspace{-1pt} = \hspace{-1pt} [[\Pi_{t+1}|\sigma_t(m_t), u_t]]$, which is equal to the RHS of \eqref{p2}.
\end{proof}

\vspace{-12pt}

\subsection{Examples of Information States} \label{subsection:info_examples}

\color{black}

In this subsection, we present examples of information states for systems with given state-space models.

\textit{1) Systems with perfectly observed states:} At each $t=0,\dots,T$, let the system's state be denoted by the uncertain variable $X_t \in \mathcal{X}_t$, where $\mathcal{X}_t$ is a known state space. The agent's observation is $Y_t = X_t$ and the agent incurs a cost $C_t = d_t(X_t, U_t)$ when they implement an action $U_t \in \mathcal{U}_t$. Starting at $X_0 \in \mathcal{X}_0$, the state evolves as $X_{t+1} = f_t(X_t, U_t, W_t)$ for all $t$. Each uncertain variable in $\{X_0, W_{t}~|~ t=0,\dots,T\}$ is independent of all other uncertain variables in that set. Then, an information state at each $t$ is $\Pi_t = X_t$, i.e., the state itself \cite{bertsekas1973sufficiently}. It takes values in the set $\mathcal{X}_t$ and satisfies \eqref{p1} - \eqref{p2} for all $t$. Note that it is always computationally advantageous to construct a DP decomposition using the state instead of the entire memory. 

\textit{2) Systems with delayed observations:} Consider the dynamics in Case 1. However, at each $t=0,\dots,T$, the agent observes the state with a delay of $n \in \mathbb{N}$ time steps, i.e., $Y_t = X_{t-n}$. Then, an information state at each $t$ is $\Pi_t = (X_{t-n}, U_{t-n:t-1})$ that takes values in the set $\mathcal{X}_t \times \prod_{t-n}^{t-1} \mathcal{U}_{\ell}$. It can be verified that this information state satisfies \eqref{p1} - \eqref{p2} for all $t$ and that it is computationally advantageous to utilize it instead of the complete memory of the agent.

\textit{3) Systems with partially observed states:} Consider the dynamics in Case 1. However, at each $t=0,\dots,T$, the agent's observation is $Y_t = h_t(X_t, N_t)$, where $N_t \in \mathcal{N}_t$ is a noise in observation and each uncertain variable in $\{X_0, W_{t}, N_t ~|~ t=0,\dots,T\}$ is independent of all other uncertain variables in that set. Then, an information state at each $t=0,\dots,T$ is the conditional range $\Pi_t = [[X_t|M_t]]$ \cite{bernhard2003minimax, Dave2021minimax}. This is a set-valued uncertain variable whose realization at time $t$ is a subset of the state space containing values of $x_t$ consistent with a given realization of the memory $m_t$. Explicitly, for a given realization of the memory $m_t \in \mathcal{M}_t$ at time $t$, the conditional range takes the realization 
$P_t := \big\{ x_t \in \mathcal{X}_t ~\big|~ \exists x_0 \in \mathcal{X}_0, w_{0:t-1} \in \prod_{\ell = 0}^{t-1} \mathcal{W}_\ell,~ n_{0:t} \in \prod_{\ell = 0}^{t} \mathcal{N}_\ell \text{ such that } y_{t} = h_{t}(x_{t},n_{t}),~ x_{\ell+1} = f_{\ell}(x_{\ell}, u_{\ell},w_{\ell}), y_{\ell} = h_{\ell}(x_{\ell}, n_{\ell}) \text{ for all } \ell = 0, \dots, t-1 \big\}.$
To establish that the conditional range is a valid information state, it is easier to verify the alternate conditions \eqref{p2a} and \eqref{p2b} instead of the property \eqref{p2} in Definition \ref{def_info_state}. Generally, it is computationally advantageous to construct a DP decomposition using the conditional range instead of the memory for systems with longer time horizons.

\color{black}


\begin{remark}
\textcolor{black}{For systems with known dynamics, any uncertain variable can be proposed as an information state using the conditions in Definition \ref{def_info_state} to simplify the DP decomposition.} However, systems with large state spaces may require further improvement in tractability, even at the cost of optimality.
Moreover, in certain applications, we need to learn a representation of the information state using incomplete knowledge of the dynamics, unable to satisfy the conditions exactly. 
In Section \ref{section:approx}, we introduce approximate information states that can address the above concerns.
\end{remark}

\vspace{-6pt}

\section{Approximate Information States} \label{section:approx}

\textcolor{black}{In this section, we define approximate information states by relaxing the conditions in Definition \ref{def_info_state}, and utilize them to develop an approximate DP decomposition which computes a sub-optimal control strategy for Problem \ref{problem_1}. In Subsection \ref{subsection:approx_properties}, we prove Lipschitz continuity of approximate value functions and guarantees on the worst-case performance of resulting approximate control strategies. Next, we present an alternate characterization of approximate information states in Subsection \ref{subsection:alt_approx}. Finally, in Subsection \ref{subsection:approx_examples}, we present examples of these performance bounds when approximate information states are constructed using state quantization.}


\begin{definition} \label{def_approx}
An \textit{approximate information state} for Problem \ref{problem_1} at each $t=0,\dots,T$ is an uncertain variable $\hat{\Pi}_t = \hat{\sigma}_t(M_t)$ taking values in a bounded set $\hat{\mathcal{P}}_t$ and generated by an $L$-invertible function $\hat{\sigma}_t : \mathcal{M}_t \to \hat{\mathcal{P}}_t$. Furthermore, for all $t=0,\dots,T$, there exist parameters $\epsilon_t, \delta_t, \lambda_t \in \mathbb{R}_{\geq0}$ such that for all $m_t \in [[M_t]]$ and $u_t \in [[U_t]]$, it satisfies the properties:

\textcolor{black}{\textit{1) Sufficient to approximately evaluate worst-case cost:}}
\begin{gather}
    \Big|\sup_{c_t \in [[C_t|m_t, u_t]]} c_t - \sup_{c_t \in [[C_t|\hat{\sigma}_t(m_t), u_t]]}  c_t\Big| \leq \epsilon_t. \label{ap1}
\end{gather}

\textcolor{black}{\textit{2) Sufficient to approximate its own evolution:}} We define the sets
$\mathcal{K}_{t+1} := [[ \hat{\Pi}_{t+1} ~|~ m_t, u_t]]$ and $\hat{\mathcal{K}}_{t+1} := [[ \hat{\Pi}_{t+1} ~|~ \hat{\sigma}_t(m_t), u_t ]],$ each independent of strategy $\boldsymbol{g}$.
Then
\begin{gather}
    \mathcal{H}(\mathcal{K}_{t+1}, \hat{\mathcal{K}}_{t+1}) \leq \delta_t, \label{ap2}
\end{gather}
where recall that $\mathcal{H}$ is the Hausdorff distance in \eqref{H_met_def}.

\textit{3) Lipschitz-like evolution:} For all $\hat{\pi}_t^1, \hat{\pi}_t^2 \in [[\hat{{\Pi}}_t]]$,
\begin{gather}
    \mathcal{H}\big([[ \hat{\Pi}_{t+1} | \hat{\pi}_t^1, u_t ]], [[ \hat{\Pi}_{t+1} |\hat{\pi}_t^2, u_t ]]\big) \leq \lambda_t \cdot \eta(\hat{\pi}_t^1, \hat{\pi}_t^2), \label{ap3}
\end{gather}
where $\eta$ is an appropriate metric on $\hat{\mathcal{P}}_t$.
\end{definition}

Using the approximate information state in Definition \ref{def_approx}, we can construct a DP as follows. For all $t$, for all $\hat{\pi}_t \in [[\hat{{\Pi}}_t]]$ and $u_t \in [[U_t]]$, we recursively define the value functions
\begin{align}
    \hat{Q}_t(\hat{\pi}_t, u_t) := &\max \Big \{ \sup_{c_t \in [[C_t|\hat{\pi}_t, u_t]]} c_t,  \nonumber \\
    & \; \; \; \sup_{\hat{\pi}_{t+1} \in [[\hat{\Pi}_{t+1}|\hat{\pi}_t, u_t]]} \hat{V}_{t+1}(\hat{\pi}_{t+1}) \Big \}, \label{DP_approx_1}\\
    \hat{V}_t(\hat{\pi}_t) := &\inf_{u_t \in [[U_t]]} \hat{Q}_t(\hat{\pi}_t, u_t), \label{DP_approx_2}
\end{align}
where $\hat{V}_{T+1}(\hat{\pi}_{T+1}) := 0$ identically. If there exists a minimizing argument in the RHS of \eqref{DP_approx_2} at each $t=0,\dots,T$, then $\hat{g}_t^*(\hat{\pi}_t) := \arg \min_{u_t \in \mathcal{U}_t} \hat{Q}_t(\hat{\pi}_t, u_t)$ constitutes an approximate control law at time $t$. Furthermore, we call  $\boldsymbol{\hat{g}}^* = \hat{g}^*_{0:T}$ an approximately optimal strategy for Problem \ref{problem_1}. In Subsection \ref{subsection:approx_properties}, we derive approximation bounds for the same. 

\begin{remark}
As we showed in Section \ref{section:info_state}, we can specialize this DP for terminal cost problems, with the value functions for all $t=0,\dots,T-1$ given by
\begin{align}
    \hat{Q}_t^\text{tm}(\hat{\pi}_t, u_t) := &\sup_{\hat{\pi}_{t+1} \in [[\hat{\Pi}_{t+1}|\hat{\pi}_t, u_t]]} \label{DP_ap_term_1} \hat{V}_{t+1}^\text{tm}(\hat{\pi}_{t+1}), \\
    \hat{V}_t^\text{tm}(\hat{\pi}_t) := &\inf_{u_t \in \mathcal{U}_t} \hat{Q}_t^\text{tm}(\hat{\pi}_t, u_t), \label{DP_ap_term_2}
\end{align}
and $\hat{Q}_T^\text{tm}(\hat{\pi}_{T},u_T) := \sup_{c_T \in [[C_T|\hat{\pi}_T, u_T]]} c_T$ and $\hat{V}_{T}^\text{tm}(\hat{\pi}_{T}):= \inf_{u_T \in \mathcal{U}_T} \hat{Q}_T^\text{tm}(\hat{\pi}_{T},u_T)$ at time $T$.
\end{remark}

\begin{remark}
\textcolor{black}{The conditions in Definition \ref{def_approx} can be investigated using only output variables. Thus, using them to construct a training loss, an approximate information state model can be learned from output data without knowledge of dynamics. This is illustrated in Subsections \ref{subsection:learning_algorithm} and \ref{subsection:pursuit}.}
\end{remark}

\begin{remark}
\textcolor{black}{The notion of an approximate information state is related to approximate bisimulations \cite{girard2007approximation}. Recall that in Section \ref{section:info_state}, we constructed a memory-based system that yielded an optimal solution to Problem \ref{problem_1}. Then, with the worst-case cost as each system's output, \eqref{p1} - \eqref{p2} in Definition \ref{def_approx} define an approximate bisimulation between an approximate system with state $\hat{\Pi}_t$ and the memory-based system with state $M_t$ at each $t$. Together, \eqref{p1} - \eqref{p2} can be used to bound the branching distance between the two systems across $t = 0,\dots,T$ \cite[Subsection VI-B]{girard2007approximation}. Instead, our subsequent results focus on the worst-case performance of approximate control strategies to facilitate approximately robust control.}
\end{remark}


\vspace{-12pt}

\subsection{Properties of Approximate Information States} \label{subsection:approx_properties}

In this subsection, we present several properties of the approximate DP \eqref{DP_approx_1} - \eqref{DP_approx_2} \textcolor{black}{using preliminary results detailed in Appendix B.} To begin, we prove in Theorem \ref{thm_L_exist} that each approximate value function is Lipschitz continuous. This property subsequently allows us to establish error bounds.


\begin{theorem} \label{thm_L_exist}
In the approximate DP \eqref{DP_approx_1} - \eqref{DP_approx_2}, the value functions $\hat{Q}_t(\hat{\pi}_t, u_t)$ and $\hat{V}_t(\hat{\pi}_t)$ are Lipschitz continuous with respect to $\hat{\pi}_t \in [[\hat{{\Pi}}_t]]$ for all $u_t \in [[{U}_t]]$ and $t=0,\dots,T$.
\end{theorem}

\begin{proof}
We prove the Lipschitz continuity of the value functions by constructing a valid candidate for the Lipschitz constant $L_{\hat{V}_t}$ at each $t=0,\dots,T$, using mathematical induction. At time $T+1$, recall that $\hat{V}_{T+1}(\hat{\pi}_{T+1}) = 0$ identically and thus, $\hat{V}_{T+1}(\hat{\pi}_{T+1})$ is trivially Lipschitz continuous with a constant $L_{\hat{V}_{T+1}} = 0$. This forms the basis of our induction. Then, at each $t=0,\dots,T$, we consider the induction hypothesis that $\hat{Q}_{t+1}(\hat{\pi}_{t+1}, u_{t+1})$ and $\hat{V}_{t+1}(\hat{\pi}_{t+1})$ are Lipschitz continuous with respect to $\hat{\pi}_{t+1} \in [[\hat{{\Pi}}_{t+1}]]$ for all $u_{t+1} \in [[{U}_{t+1}]]$, and denote the constant by $L_{\hat{V}_{t+1}} \in \mathbb{R}_{\geq0}$. 

At time $t$, we first prove the result for the value function $\hat{Q}_t(\hat{\pi}_t, u_t)$. Let $\hat{\pi}_t^1, \hat{\pi}_t^2 \in [[\hat{{\Pi}}_t]]$ be two possible realizations of $\hat{\Pi}_t$. Then, using the definition \eqref{DP_approx_1} of $\hat{Q}_t(\hat{\pi}_t, u_t)$ and \eqref{eq_prelim_3} from Lemma \ref{lem_prelim_3}, we state that
\begin{multline}
|\hat{Q}_t(\hat{\pi}_t^1, u_t) -  \hat{Q}_t(\hat{\pi}_t^2, u_t)| \leq \max \Big \{ \Big | \sup_{c_t^1 \in [[C_t|\hat{\pi}_t^1, u_t]]} c_t^1 \\
- \sup_{c_t^2 \in [[C_t|\hat{\pi}_t^2, u_t]]} c_t^2\Big|, \Big| \sup_{\hat{\pi}_{t+1}^1 \in [[\hat{\Pi}_{t+1}|\hat{\pi}_t^1, u_t]]}\hat{V}_{t+1}(\hat{\pi}_{t+1}^1) \\
- \sup_{\hat{\pi}_{t+1}^2 \in [[\hat{\Pi}_{t+1}|\hat{\pi}_t^2, u_t]]}\hat{V}_{t+1}(\hat{\pi}_{t+1}^2) \Big | \Big \}. \label{eq_exist_1}
\end{multline}
We consider the RHS of \eqref{eq_exist_1} term by term. In the first term, we note that for all $\hat{\pi}_t \in [[\hat{{\Pi}}_t]]$,
\begin{gather}
   \sup_{c_t \in [[C_t|\hat{\pi}_t, u_t]]}  c_t =  \sup_{m_t \in [[M_t|\hat{\pi}_t]]} \Big(\sup_{c_t \in [[C_t|m_t, u_t]]} c_t\Big). \label{eq_exist_4}
\end{gather}
In the RHS of \eqref{eq_exist_4}, recall from Assumption \ref{assumption_2} that the uncertain variable $C_t$ is a Lipschitz function of $(W_{0:t}, U_{0:t})$, and $(M_t, U_t)$ is an $L$-invertible function of $(W_{0:t}, U_{0:t})$. Thus, using \eqref{eq_range_lipschitz} from Lemma \ref{lem_range_lipschitz}, there exists a constant $L_{C|M,U}$ such that $\mathcal{H}([[C_t|m_t^1, u_t]], [[C_t|m_t^2, u_t]]) \leq L_{M|C,U} \cdot \eta(m_t^1, m_t^2)$ for all $m^1, m^2 \in [[M_t]]$. Furthermore, we use \eqref{eq_prelim_5} to state that
\begin{multline}
    \Big|\sup_{c_t^1 \in [[C_t|m_t^1, u_t]]} c_t^1 - \sup_{c_t^2 \in [[C_t|m_t^2, u_t]]}
    c_t^2\Big| \\
    \leq L_{M|C,U}\cdot L_{c_t} \cdot \eta(m_t^1,m_t^2). \label{eq_exist_5}
\end{multline}
Then, consider a function $e_t:\mathcal{M}_t \times \mathcal{U}_t \to \mathbb{R}_{\geq0}$ defined as $e_t(m_t, u_t) := \sup_{c_t \in [[C_t|m_t, u_t]]} c_t$. As a direct consequence of \eqref{eq_exist_5}, $e_t$ is Lipschitz continuous with respect to $m_t$ with a constant $L_{e_t} := L_{M|C,U}\cdot L_{c_t}$. Using \eqref{eq_exist_4} and the definition of $e_t$ in the first term in the RHS of \eqref{eq_exist_1},
\begin{multline}
    \Big | \sup_{c_t^1 \in [[C_t|\hat{\pi}_t^1, u_t]]} c_t^1 - \sup_{c_t^2 \in [[C_t|\hat{\pi}_t^2, u_t]]} c_t^2\Big| \\
    = \Big | \sup_{m_t^1 \in [[M_t|\hat{\pi}_t^1]]} e_t(m_t^1, u_t) - \sup_{m_t^2 \in [[M_t|\hat{\pi}_t^2]]} e_t(m_t^2, u_t) \Big |. \label{eq_exist_6}
\end{multline}
In \eqref{eq_exist_6}, recall that the uncertain variable $\hat{\Pi}_t$ is an $L$-invertible function of $M_t$ and thus, the conditional range $[[M_t|\hat{\pi}_t]]$ satisfies \eqref{L_inv_2}. Then, we use \eqref{eq_prelim_5} once more to state that
\begin{gather}
    \hspace{-5pt} \Big | \hspace{-2pt} \sup_{c_t^1 \in [[C_t|\hat{\pi}_t^1, u_t]]} \hspace{-5pt} c_t^1 - \hspace{-5pt} \sup_{c_t^2 \in [[C_t|\hat{\pi}_t^2, u_t]]} \hspace{-5pt} c_t^2 \hspace{-1pt} \Big| \leq  L_{\hat{\sigma}_t^{-1}} \hspace{-2pt} \cdot \hspace{-2pt} L_{e_t} \hspace{-2pt} \cdot  \hspace{-2pt} \eta(\hat{\pi}_t^1,\hat{\pi}_t^2). \label{eq_exist_2}
\end{gather}
In the second term in the RHS of \eqref{eq_exist_1}, we use the induction hypothesis and \eqref{eq_prelim} from Lemma \ref{lem_prelim} to write that
\begin{multline}
    \hspace{-5pt} \Big| \sup_{\hat{\pi}_{t+1}^1 \in [[\hat{\Pi}_{t+1}|\hat{\pi}_t^1, u_t]]} \hspace{-3pt} \hat{V}_{t+1}(\hat{\pi}_{t+1}^1) - \sup_{\hat{\pi}_{t+1}^2 \in [[\hat{\Pi}_{t+1}|\hat{\pi}_t^2, u_t]]} \hspace{-3pt} \hat{V}_{t+1}(\hat{\pi}_{t+1}^2) \Big | \\
    \leq L_{\hat{V}_{t+1}}\cdot \mathcal{H}\big([[\hat{\Pi}_{t+1}|\hat{\pi}_t^1, u_t]], [[\hat{\Pi}_{t+1}|\hat{\pi}_t^2, u_t]] \big) \\
    \leq L_{\hat{V}_{t+1}}\cdot \lambda_{t} \cdot \eta(\hat{\pi}_t^1, \hat{\pi}_t^2), \label{eq_exist_3}
\end{multline}
where, in the second inequality, we use the third property \eqref{ap3} of approximate information states in Definition \ref{def_approx}. Then, the proof for $\hat{Q}_t(\hat{\pi}_t, u_t)$ is complete by substituting \eqref{eq_exist_2} and  \eqref{eq_exist_3} into the RHS of \eqref{eq_exist_1} and \color{black} defining 
\begin{gather} \label{L_hatV_def}
    L_{\hat{Q}_t} := \max \big\{L_{\hat{\sigma}_t^{-1}} \cdot L_{e_t}, L_{\hat{V}_{t+1}} \cdot \lambda_{t} \big\}.
\end{gather} \color{black}
To prove the result for $\hat{V}_t(\hat{\pi}_t)$, we use \eqref{prelim_2_2} from Lemma \ref{lem_prelim_2} to state that $\big|\hat{V}_t(\hat{\pi}_t^1) - \hat{V}_t(\hat{\pi}_t^2)\big| = \big|\inf_{u_t \in [[U_t]]}\hat{Q}_t(\hat{\pi}_t^1, u_t) -  \inf_{u_t \in [[U_t]]}\hat{Q}_t(\hat{\pi}_t^2, u_t)\big|$
    $\leq \sup_{u_t \in [[U_t]]} \big|\hat{Q}_t(\hat{\pi}_t^1, u_t) - \hat{Q}_t(\hat{\pi}_t^2, u_t) \big| \leq L_{\hat{Q}_t} \cdot \eta(\hat{\pi}_t^1,\hat{\pi}_t^2),$
which proves the induction hypothesis at time $t$. Thus, the result holds using mathematical induction.
\end{proof}

\begin{remark}
    \textcolor{black}{From \eqref{L_hatV_def}, we can construct $L_{\hat{Q}_t}$ and $L_{\hat{V}_t}$ at any $t$ as functions of the Lipschitz constants of the underlying dynamics, the $L$-invertible constants $(L_{\hat{\sigma}^{-1}_t}, \dots, L_{\hat{\sigma}^{-1}_t})$, and $(\lambda_t, \dots, \lambda_T)$ from \eqref{ap3}. The Lipschitz constant $L_{\hat{V}_t}$ is not an explicit function of the bounds $\epsilon_t$ in \eqref{ap1} and $\delta_t$ in \eqref{ap2}. Furthermore, a smaller value of $L_{\hat{V}_t}$ can be found in special cases, such as problems where $\hat{\mathcal{P}}_t$ and $\mathcal{U}_t$ are finite for all $t$.}
\end{remark}

Next, we establish an upper bound on the approximation error when the value functions of the optimal DP \eqref{DP_basic_1} - \eqref{DP_basic_2} are estimated using the approximate DP \eqref{DP_approx_1} - \eqref{DP_approx_2} at each $t$.

\begin{theorem} \label{thm_approx_term_dp}

Let $L_{\hat{V}_{t+1}}$ be the Lipschitz constant of $\hat{V}_{t+1}$ for all $t=0,\dots,T$. Then, for all $m_t \in [[M_t]]$ and $u_t \in [[U_t]]$,
\begin{align}
    |Q_t(m_t,u_t) - \hat{Q}_t(\hat{\sigma}_t(m_t),u_t)| &\leq \alpha_t, \label{thm_3_1} \\
    |V_t(m_t) - \hat{V}_t(\hat{\sigma}_t(m_t))| &\leq \alpha_t, \label{thm_3_2}
\end{align}
where $\alpha_t = \max(\epsilon_t, \alpha_{t+1} + L_{\hat{V}_{t+1}}\cdot \delta_t)$ for all $t$, with $\alpha_{T+1} = 0$.
\end{theorem}

\begin{proof}
For all $t=0,\dots,T$, let $m_t \in [[M_t]]$ and $u_t \in [[U_t]]$ be realizations of $M_t$ and $U_t$, respectively.
We prove both results by mathematical induction, starting with time step $T+1$. At $T+1$, by definition, $V_{T+1}(m_{T+1}, u_{T+1}) = V_{T+1}(\hat{\sigma}_{T+1}(m_{T+1})) = 0$. This forms the basis of our mathematical induction. Then, at each $t=0,\dots,T$, we consider the induction hypothesis $|V_{t+1}(m_{t+1}) - \hat{V}_{t+1}(\hat{\sigma}_{t+1}(m_{t+1}))| \leq \alpha_{t+1}$.
At time $t$, we first prove \eqref{thm_3_1}. Using \eqref{eq_prelim_3} from Lemma \ref{lem_prelim_3} in the LHS of \eqref{thm_3_1} to state that
\begin{align}
    &|Q_t(m_t,u_t) - \hat{Q}_t(\hat{\sigma}_t(m_t),u_t)| \leq \max \Big\{ \Big| \hspace{-1pt} \sup_{c_t \in [[C_t|m_t, u_t]]}  c_t \nonumber \\
    &- \hspace{-5pt}  \sup_{c_t \in [[C_t|\hat{\sigma}_t(m_t), u_t]]}c_t\Big|, \Big| \sup_{m_{t+1} \in [[M_{t+1}|m_t, u_t]]} {V}_{t+1}(m_{t+1}) \nonumber \\ 
    & - \hspace{-4pt} \sup_{\hat{\pi}_{t+1} \in [[\hat{\Pi}_{t+1}|\hat{\sigma}_t(m_t), u_t]]} \hat{V}_{t+1}(\hat{\pi}_{t+1})\Big| \Big\}. \label{proof_3_1}
\end{align}
We consider the RHS of \eqref{proof_3_1} term-by-term.
By direct application of \eqref{ap1} in Definition \ref{def_approx}, the first term in the RHS satisfies 
\begin{gather}
    \Big|\sup_{c_t \in [[C_t|m_t, u_t]]}c_t -  \sup_{c_t \in [[C_t|\hat{\sigma}_t(m_t), u_t]]}c_t\Big| \leq \epsilon_t. \label{proof_3_2}
\end{gather}
For the second term in the RHS of \eqref{proof_3_1}, we use the triangle inequality to write that
\begin{align}
    &\Big| \sup_{m_{t+1} \in [[M_{t+1}|m_t, u_t]]} \hspace{-5pt} {V}_{t+1}(m_{t+1})- \hspace{-5pt} \sup_{\hat{\pi}_{t+1} \in [[\hat{\Pi}_{t+1}|\hat{\sigma}_t(m_t), u_t]]} \nonumber \\
    &\hat{V}_{t+1}(\hat{\pi}_{t+1})\Big| \leq \Big| \sup_{m_{t+1} \in [[M_{t+1}|m_t, u_t]]} V_{t+1}(m_{t+1})  \nonumber \\
    & - \hspace{-5pt} \sup_{\hat{\sigma}_{t+1}(m_{t+1}) \in [[\hat{\Pi}_{t+1}|m_t, u_t]]} \hspace{-7pt} \hat{V}_{t+1}(\hat{\sigma}_{t+1}(m_{t+1})) \Big| + \nonumber \\
    &\Big|\hspace{-2pt} \sup_{\hat{\pi}_{t+1} \in [[\hat{\Pi}_{t+1}|m_t, u_t]]} \hspace{-12pt} \hat{V}_{t+1}(\hat{\pi}_{t+1}) 
    - \hspace{-7pt} \sup_{\hat{\pi}_{t+1} \in [[\hat{\Pi}_{t+1}|\hat{\sigma}_t(m_t), u_t]]} \hspace{-19pt} \hat{V}_{t+1}(\hat{\pi}_{t+1}) \Big|. \label{proof_3_3}
\end{align}
For the first term in the RHS of \eqref{proof_3_3}, we first note that $\sup_{\hat{\sigma}_{t+1}(m_{t+1}) \in [[\hat{\Pi}_{t+1}|m_t, u_t]]} \hat{V}_{t+1}(\hat{\sigma}_{t+1}(m_{t+1})) = \sup_{m_{t+1} \in [[M_{t+1}|m_t, u_t]]}  \hat{V}_{t+1}(\hat{\sigma}_{t+1}(m_{t+1}))$ 
because $[[\hat{\Pi}_{t+1}$ $|~m_t, u_t]]= \{\hat{\sigma}_{t+1}(m_{t+1}) \in \hat{\mathcal{P}}_t ~|~ m_{t+1} \in[[M_{t+1}~|~m_t, u_t]] \}$. 
Then, we can state that $\big| \sup_{m_{t+1} \in [[M_{t+1}|m_t, u_t]]}$ $ V_{t+1}(m_{t+1}) \hspace{-1pt} - \hspace{-1pt} \sup_{\hat{\sigma}_{t+1}(m_{t+1}) \in [[\hat{\Pi}_{t+1}|m_t, u_t]]} \hspace{-2pt} \hat{V}_{t+1}(\hat{\sigma}_{t+1}(m_{t+1})) \big|$ $ \leq \sup_{m_{t+1} \in [[M_{t+1}|m_t, u_t]]} \big|V_{t+1}(m_{t+1})
     - \hat{V}_{t+1}(\hat{\sigma}_{t+1}(m_{t+1}))\big|$ $
    \leq \alpha_{t+1},$
where, in the first inequality, we use \eqref{prelim_2_1} from Lemma \ref{lem_prelim_2}; and, in the second inequality, we use the induction hypothesis for time $t+1$. 
The second term in the RHS of \eqref{proof_3_3} satisfies $\big|\sup_{\hat{\pi}_{t+1} \in [[\hat{\Pi}_{t+1}|m_t, u_t]]} \hat{V}_{t+1}(\hat{\pi}_{t+1}) 
    - \hspace{-5pt} \sup_{\hat{\pi}_{t+1} \in [[\hat{\Pi}_{t+1}|\hat{\sigma}_t(m_t), u_t]]} \hat{V}_{t+1}(\hat{\pi}_{t+1}) \big| \leq L_{\hat{V}_{t+1}} \hspace{-4pt} \cdot \delta_t$ using \eqref{eq_prelim} from Lemma \ref{lem_prelim} and \eqref{ap2} from Definition \ref{def_approx}.
Substituting the respective inequalities for each term in the RHS of \eqref{proof_3_3} yields $\big| \sup_{m_{t+1} \in [[M_{t+1}|m_t, u_t]]}  {V}_{t+1}(m_{t+1})-  \sup_{\hat{\pi}_{t+1} \in [[\hat{\Pi}_{t+1}|\hat{\sigma}_t(m_t), u_t]]}$ 
    $\hat{V}_{t+1}(\hat{\pi}_{t+1})\big| \leq \alpha_{t+1} + L_{\hat{V}_{t+1}} \cdot \delta_t$.
We complete the proof for \eqref{thm_3_1} by substituting the inequalities in the RHS of \eqref{proof_3_2} and \eqref{proof_3_3} into the RHS of \eqref{proof_3_1}. Next, we prove  \eqref{thm_3_2} at time $t$. Using the definition of the value functions in the LHS of \eqref{thm_3_2}, we write that
\begin{multline}
    |V_t(m_t) - \hat{V}_t(\hat{\sigma}_t(m_t))| = \Big|\inf_{u_t \in [[U_t]]}Q_t(m_t, u_t) - \inf_{u_t \in [[U_t]]} \\
    \hat{Q}_t(\hat{\sigma}_t(m_t), u_t) \Big|
    \leq \sup_{u_t \in [[U_t]]}|Q_t(m_t, u_t) - \hat{Q}_t(\hat{\sigma}_t(m_t), u_t)| \\
    \leq \max\{\epsilon_t, \alpha_{t+1} + L_{\hat{V}_{t+1}} \cdot \delta_t\},
\end{multline}
where in the first inequality, we use \eqref{prelim_2_2} from Lemma \ref{lem_prelim_2}; and in the second inequality, we use \eqref{thm_3_1}. Thus, the results hold for all $t=0,\dots,T$ using mathematical induction.
\end{proof}

After bounding the approximation error for value functions, we also seek to bound the maximum performance loss in the implementation of an approximately optimal strategy. Consider an approximate strategy $\boldsymbol{\hat{g}}^* := \hat{g}_{0:T}^*$ computed using \eqref{DP_approx_1} - \eqref{DP_approx_2}, where $\hat{g}_t^*(\hat{\pi}_t) = \arg \min_{u_t \in [[U_t]]} \hat{Q}_t(\hat{\pi}_t, u_t)$ for all $t$. 
We can construct an approximate memory-based strategy $\boldsymbol{g}^{\text{ap}}=g_{0:T}^{\text{ap}}$ by selecting the control law $g_t^{\text{ap}}(m_t) := \hat{g}_t^*(\hat{\sigma}_t(m_t))$ for all $t$. Note that $\boldsymbol{g}^{\text{ap}}$ is equivalent to $\hat{\boldsymbol{g}}^*$ because they generate the same actions at each $t$ and subsequently, yield the same performance. Thus, we evaluate the performance of $\boldsymbol{g}^{\text{ap}}$ to determine the quality of approximation. To this end, for all $t=0,\dots,T$, for all $m_t \in [[M_t]]$ and $u_t \in [[U_t]]$, we define
\begin{align}
\Theta_t(m_t,u_t) :=  &\max\Big\{ \sup_{c_t \in [[C_t|m_t, u_t]]}c_t, \nonumber \\
&\sup_{m_{t+1} \in  [[M_{t+1}|m_t, u_t]]}\Lambda_{t+1}(m_{t+1}) \Big\}, \label{value_g_1}\\
\Lambda_t(m_t) := &\Theta_t(m_t,g_t^{\text{ap}}(m_t)), \label{value_g_2} 
\end{align}
where $\Lambda_{T+1}(m_{T+1}) := 0$, identically. Then, the performance of the memory-based approximate strategy $\boldsymbol{g}^{\text{ap}}$ is $\Lambda_0(m_0)$. In contrast, recall that the performance of an optimal strategy $\boldsymbol{g}^*$ is the optimal value $V_0(m_0)$ computed using \eqref{DP_basic_1} - \eqref{DP_basic_2}. Next, we bound the difference in performance between $\boldsymbol{g}^{\text{ap}}$ and $\boldsymbol{g}^*$.

\begin{theorem} \label{thm_approx_term_policy}
Let $L_{\hat{V}_{t+1}}$ be the Lipschitz constant of $\hat{V}_{t+1}$ for all $t=0,\dots,T$. Then, for all $m_t \in [[M_t]]$ and $u_t \in [[U_t]]$,
\begin{align}
    |Q_t(m_t, u_t) - \Theta_t(m_t, u_t)| \leq 2\alpha_t, \label{thm_3_3} \\
    |V_t(m_t) - \Lambda_t(m_t)| \leq 2\alpha_t. \label{thm_3_4}
\end{align}
where $\alpha_t = \max(\epsilon_t, \alpha_{t+1} + L_{\hat{V}_{t+1}}\cdot \delta_t)$ for all $t$, with $\alpha_{T+1} = 0$.
\end{theorem}

\begin{proof}
We begin by recursively defining the value functions that compute the performance of the strategy $\boldsymbol{\hat{g}}$. 
For all $t=0,\dots,T$ and for each $\hat{\pi}_t \in [[\hat{\Pi}_t]]$ and $u_t \in [[U_t]]$, let
\begin{align}
    \hat{\Theta}_t (\hat{\pi}_t, u_t) := &\max \Big \{ \sup_{c_t \in [[C_t|\hat{\pi}_t,u_t]]} c_t, \nonumber \\
    &\sup_{\hat{\pi}_{t+1} \in [[\hat{\Pi}_{t+1}|\hat{\pi}_t, u_t]]} \hat{\Lambda}_{t+1}(\hat{\pi}_{t+1})\Big\}, \label{value_hat_g_1}\\
    \hat{\Lambda}_t(\hat{\pi}_t):= &\hat{\Theta}_t(\hat{\pi}_t, \hat{g}_t(\hat{\pi}_t)), \label{value_hat_g_2}
\end{align}
where $\hat{\Lambda}_{T+1}(\hat{\pi}_{T+1}) :=0$, identically. Note that
\begin{align}
    \hat{\Theta}_t(\hat{\pi}_t, u_t) = \hat{Q}_{t}(\hat{\pi}_t, u_t) \quad \text{and} \quad \hat{\Lambda}_t(\hat{\pi}_t) = \hat{V}_t(\hat{\pi}_t), \label{proof_4_1}
\end{align}
for all $t=0,\dots,T$, since $\hat{g}_t(\hat{\pi}_t)  =\arg\min_{u_t \in \mathcal{U}_t}\hat{Q}_t(\hat{\pi}_t,u_t)$.

We first prove \eqref{thm_3_3} for all $t=0,\dots,T$. At time $t$, using the triangle inequality and \eqref{proof_4_1} in the LHS of \eqref{thm_3_3}:
\begin{align}
    |Q_t(m_t, u_t) - &\Theta_t(m_t, u_t)| \leq |Q_t(m_t, u_t) -  \nonumber \\
    \hat{Q}_t(\hat{\sigma}_t(m_t),u_t)|
     &+ |\hat{\Theta}_t(\hat{\sigma}_t(m_t), u_t) - \Theta_t(m_t, u_t)| \nonumber \\
    \leq \alpha_t  + &|\hat{\Theta}_t(\hat{\sigma}_t(m_t), u_t) - \Theta_t(m_t, u_t)|, \label{proof_4_2}
\end{align}
where, in the second inequality, we use \eqref{thm_3_1} from Theorem \ref{thm_approx_term_dp}. Then, to prove \eqref{thm_3_3}, it suffices to show that
\begin{gather}
|\hat{\Theta}_t(\hat{\sigma}_t(m_t), u_t) - \Theta_t(m_t, u_t)| \leq \alpha_t. \label{proof_4_3}
\end{gather}
We use mathematical induction starting at time $T+1$ to prove \eqref{proof_4_3} in addition to $|\hat{\Lambda}_t(\hat{\sigma}_t(m_t)) - \Lambda_t(m_t)| \leq \alpha_t$ for all $t=0,\dots,T$. At time $T+1$, using the definitions it holds that $\hat{\Lambda}_{T+1}(\hat{\sigma}_{T+1}(m_{T+1})) = \Lambda_{T+1}(m_{T+1}) = 0$. This forms the basis of our induction. Next, for all $t=0,\dots,T$, we consider the induction hypothesis that $|\hat{\Lambda}_{t+1}(\hat{\sigma}_{t+1}(m_{t+1})) - \Lambda_{t+1}(m_{t+1})| \leq \alpha_{t+1}$. Given the hypothesis, \eqref{proof_4_3} holds at time $t$ using the same sequence of arguments as in the proof for Theorem \ref{thm_approx_term_dp}. 
Next, using the definitions of the value functions from \eqref{value_g_2} and \eqref{value_hat_g_2}, we write that
\begin{multline}
    \hspace{-10pt}|\hat{\Lambda}_t(\hat{\sigma}_t(m_t)) - \Lambda_t(m_t)| = | \hat{\Theta}_t(\hat{\sigma}_t(m_t), \hat{g}_t(\hat{\sigma}_t(m_t)) - \Theta_t(m_t,\\
    \hspace{-8pt} g_t(m_t))|= | \hat{\Theta}_t(\hat{\sigma}_t(m_t), \hat{u}_t) \hspace{-2pt} -  \hspace{-2pt} \Theta_t(m_t,\hat{u}_t)| \leq \alpha_t, \label{proof_4_7}
\end{multline}
where, in the second equality, we use the definition of the control law to write that $g_t(m_t) = \hat{g}_t(\hat{\sigma}_t(m_t)) =: \hat{u}_t$; and in the inequality, we use \eqref{proof_4_3}. This proves the induction hypothesis for time $t$ given the hypothesis for time $t+1$. Thus, using mathematical induction \eqref{proof_4_3} holds for all $t=0,\dots,T$. Subsequently, we complete the proof for \eqref{thm_3_3} for all $t=0,\dots,T$ by substituting \eqref{proof_4_3} into the RHS of \eqref{proof_4_2}. Furthermore, note that \eqref{thm_3_4} follows directly from \eqref{thm_3_3} using the same sequence of arguments used to prove \eqref{proof_4_7}.
\end{proof}

\begin{remark}
    \textcolor{black}{The definition of $\alpha_t$ at any $t$ characterizes a trade-off in approximate information states between cost and observation approximation. In this definition, the observation prediction error $\delta_t$ is scaled by $L_{\hat{V}_t}$ and accumulated with $\alpha_{t+1}$. Thus, to minimize $\alpha_t$, we typically prioritize a reduction in $\delta_t$ over a reduction in the cost approximation error $\epsilon_t$.}
\end{remark}

\begin{remark}
We can specialize the results of both Theorem \ref{thm_approx_term_dp} and Theorem \ref{thm_approx_term_policy} to terminal cost problems, where the optimal DP is given by \eqref{DP_ad_term_1} - \eqref{DP_ad_term_2} and the approximate DP is given by \eqref{DP_ap_term_1} - \eqref{DP_ap_term_2}. Then, the approximation bounds hold with $\alpha_t := \alpha_{t+1} + L_{\hat{V}_{t+1}^{\text{tm}}}\cdot \delta_t$ for all $t=0,\dots,T-1$ and $\alpha_T := \epsilon_T$.
\end{remark}

\vspace{-12pt}

\subsection{Alternate Characterization} \label{subsection:alt_approx}

In this subsection, we provide stronger but simpler conditions to identify an approximate information state as alternatives to \eqref{ap2} and \eqref{ap3}. They prescribe that an approximate information state $\hat{\Pi}_t = \hat{\sigma}_t(M_t)$ must satisfy for all $t$:

\textit{1) State-like evolution:} There exists a Lipschitz continuous function $\hat{f}_t:\hat{\mathcal{P}}_t \times \mathcal{U}_t \times \mathcal{Y}_{t+1} \to \mathcal{P}_{t+1}$, such that
\begin{gather}
    \hat{\Pi}_{t+1} = \hat{f}_t(\hat{\Pi}_t, U_t, Y_{t+1}). \label{ap2a}
\end{gather}

\textit{2) Sufficient to approximate observations:} For all $m_t \in [[M_t]]$ and $u_t \in [[U_t]]$, we define the sets $\mathcal{K}_{t+1}^{\text{ob}} := [[ Y_{t+1} ~|~ m_t, u_t]]$ and $\hat{\mathcal{K}}_{t+1}^{\text{ob}} := [[ Y_{t+1} ~|~ \hat{\sigma}_t(m_t), u_t ]]$, each independent of the strategy $\boldsymbol{g}$.
Then, for some $\delta_t^{\text{ob}} \in \mathbb{R}_{\geq0}$:
\begin{gather}
    \mathcal{H}(\mathcal{K}_{t+1}^{\text{ob}}, \hat{\mathcal{K}}_{t+1}^{\text{ob}}) \leq \delta_t^{\text{ob}}. \label{ap2b}
\end{gather}

\textit{3) Lipschitz-like observation prediction:} There exists a constant $\lambda_{t}^{\text{ob}} \in \mathbb{R}_{\geq0}$ such that for all $\hat{\pi}_t^1, \hat{\pi}_t^2 \in [[\hat{\Pi}_t]]$,
\begin{gather}
    \mathcal{H}\big([[ Y_{t+1} | \hat{\pi}_t^1, u_t ]], [[ Y_{t+1} |\hat{\pi}_t^2, u_t ]]\big) \leq \lambda_{t}^{\text{ob}} \cdot \eta(\hat{\pi}_t^1, \hat{\pi}_t^2), \label{ap2c}
\end{gather}
where $\eta$ is an appropriate metric on $\hat{\mathcal{P}}_t$.

In addition to \eqref{ap1} in Definition \ref{def_approx}, the conditions \eqref{ap2a} - \eqref{ap2c} are sufficient to characterize an approximate information state. 

\begin{lemma} \label{lem_alt_approx}
For all $t=0,\dots,T$, if an uncertain variable $\hat{\Pi}_t = \hat{\sigma}_t(M_t)$ satisfies \eqref{ap2a} - \eqref{ap2b}, it also satisfies \eqref{ap2}.
\end{lemma}

\begin{proof}
Let $m_t \in [[M_t]]$ be a given realization of $M_t$ and let $\hat{\pi}_t = \hat{\sigma}_t(m_t)$ satisfy \eqref{ap2a} - \eqref{ap2b}, for all $t=0,\dots,T$. Then, using \eqref{ap2a}, we can write the LHS in \eqref{ap2} as
$\mathcal{H}(\mathcal{K}_{t+1}, \hat{\mathcal{K}}_{t+1})     = \mathcal{H}\big([[\hat{f}_t(\hat{\sigma}_t(m_t), u_t, Y_{t+1})|m_t, u_t]], [[\hat{f}_t(\hat{\sigma}_t(m_t), u_t, Y_{t+1})|$ $\hat{\sigma}_t(m_t), u_t]] \big) = \max \big\{ \sup_{y_{t+1} \in \mathcal{K}^\text{ob}_{t+1}}  \inf_{\hat{y}_{t+1} \in \hat{\mathcal{K}}^\text{ob}_{t+1}}$ $\eta \big(\hat{f}_t(\hat{\sigma}_t(m_t),u_t,y_{t+1}), \hat{f}_t(\hat{\sigma}_t(m_t),u_t,\hat{y}_{t+1})\big), \sup_{\hat{y}_{t+1} \in \hat{\mathcal{K}}^\text{ob}_{t+1}}$ $\inf_{y_{t+1} \in \mathcal{K}^{\text{ob}}_{t+1}} \eta\big(\hat{f}_t(\hat{\sigma}_t(m_t),u_t,y_{t+1}),\hat{f}_t(\hat{\sigma}_t(m_t),u_t,\hat{y}_{t+1})\big)\big\},$ where, in the second equality, we use the definition of the Hausdorff distance from \eqref{H_met_def}. Note that $\hat{f}_t$ is globally Lipschitz from the alternate characterization of the approximate information state. This implies that $\eta\big(\hat{f}_t(\hat{\sigma}_t(m_t),u_t,y_{t+1}),$ $\hat{f}_t(\hat{\sigma}_t(m_t),u_t,\hat{y}_{t+1})\big) \leq L_{\hat{f}_t} \cdot \eta(y_{t+1},$ $\hat{y}_{t+1})$, and thus 
$\mathcal{H}(\mathcal{K}_{t+1}, \hat{\mathcal{K}}_{t+1}) \leq L_{\hat{f}_t} \max\big\{ \sup_{y_{t+1} \in \mathcal{K}^{\text{ob}}_{t+1}}$ $\inf_{\hat{y}_{t+1} \in \hat{\mathcal{K}}^\text{ob}_{t+1}} \eta(y_{t+1}, \hat{y}_{t+1}),
     \sup_{\hat{y}_{t+1} \in \hat{\mathcal{K}}^\text{ob}_{t+1}}\inf_{y_{t+1} \in \mathcal{K}^{\text{ob}}_{t+1}}$ $\eta(y_{t+1},$ $\hat{y}_{t+1})\big\}
    = L_{\hat{f}_t} \cdot \mathcal{H}(\mathcal{K}^\text{ob}_{t+1},\hat{\mathcal{K}}^\text{ob}_{t+1}) \leq L_{\hat{f}_t} \cdot \delta_t^{\text{ob}}$.
\end{proof}

\begin{lemma} \label{lem_alt_approx_2}
For all $t=0,\dots,T$, if an uncertain variable $\hat{\Pi}_t = \hat{\sigma}_t(M_t)$ satisfies \eqref{ap2a} - \eqref{ap2c}, it also satisfies \eqref{ap3}.
\end{lemma}

\begin{proof}
Let $\hat{\pi}_t^1, \hat{\pi}_t^2\in [[\hat{\Pi}_t]]$ be two possible realizations of an approximate information state $\hat{\Pi}_t$, which satisfies \eqref{ap2a} - \eqref{ap2c}, for all $t=0,\dots,T$. Then, using \eqref{ap2a}, we can write the LHS in \eqref{ap3} as
$\mathcal{H}\big([[ \Pi_{t+1} | \hat{\pi}_t^1, u_t ]], [[ \Pi_{t+1} |\hat{\pi}_t^2, u_t ]]\big)$
    $= \mathcal{H}\big([[\hat{f}_t(\hat{\pi}_t^1, u_t, Y_{t+1})|\hat{\pi}_t^1, u_t]], [[\hat{f}_t(\hat{\pi}_t^2, u_t, Y_{t+1})| \hat{\pi}_t^2, u_t]]$
    $\leq L_{\hat{f}_t} \cdot \big( \eta(\hat{\pi}_t^1, \hat{\pi}_t^2) + \mathcal{H}\big([[ Y_{t+1} | \hat{\pi}_t^1, u_t ]], [[ Y_{t+1} |\hat{\pi}_t^2, u_t ]]\big) \big)$
    $\leq L_{\hat{f}_t} \cdot (1 + \lambda_t^{\text{ob}}) \cdot \eta(\hat{\pi}_t^1, \hat{\pi}_t^2)$,
where, in the first inequality, we use the Lipschitz continuity of the function $\hat{f}_t$ along with the triangle inequality, and in the second inequality, we use \eqref{ap2c}. This completes the proof by defining $\lambda_t := L_{\hat{f}_t} \cdot (1 + \lambda_t^{\text{ob}})$.
\end{proof}

\vspace{-12pt}

\subsection{Examples} \label{subsection:approx_examples}

\textcolor{black}{We present two examples approximate information states constructed using state-quantization \cite{bertsekas1975convergence} in systems with a known state-space model. Consider a system as described in Subsection \ref{subsection:info_examples} with compact feasible sets $\big\{\mathcal{X}_t, \mathcal{N}_t, \mathcal{W}_t ~|~ t=0,\dots,T\big\}$ in a metric space $(\mathcal{S}, \eta)$. Recall that $\mathcal{X}_t$ is the state space at any $t$.
A \textit{quantized state space} is defined as a finite set $\hat{\mathcal{X}}_t \subset \mathcal{X}_t$, such that for a given quantization parameter $\gamma_t \in \mathbb{R}_{\geq0}$, it holds that $\max_{x_t \in \mathcal{X}_t} \min_{\hat{x}_t \in \hat{\mathcal{X}}_t} \eta(x_t, \hat{x}_t) \leq \gamma_t$. This implies that for every point $x_t \in \mathcal{X}_t$, there exists a point in the quantized space $\hat{X}_t$ that is at most $\gamma_t$ distance away from $x_t$. 
Then, a \textit{quantization function} $\mu_t: \mathcal{X}_t \to \hat{\mathcal{X}}_t$ maps each point in the state space to the closest point in the quantized space for all $t=0,\dots,T$, and is defined as $\mu_t(x_t) := \arg \min_{\hat{x}_t \in \hat{\mathcal{X}}_t} \eta(x_t, \hat{x}_t)$ for every $x_t \in \mathcal{X}_t$.}


\textit{1) Perfectly Observed Systems:} Consider a system where $Y_t = X_t$ for all $t=0,\dots,T$. Recall from Subsection \ref{subsection:info_examples} that the $\Pi_t=X_t \in \mathcal{X}_t$ for all $t$. \textcolor{black}{Building upon this information state, an approximate information state for such a system can be constructed using the state quantization function, i.e., $\hat{\Pi}_t := \mu_t(X_t)$.} This approximate information state Definition \ref{def_approx} with $\epsilon_t = 2L_{d_t} \cdot \gamma_t$ and $\delta_t = 2 \gamma_{t+1} + 2  L_{f_t} \cdot \gamma_t$, where $\gamma_{T+1} = 0$, and $L_{d_t}$ and $L_{f_t}$ are the Lipschitz constants for $d_t$ and $f_t$, respectively (proof in Appendix C). Note that because $\hat{\Pi}_t$ takes values in a finite set, it trivially satisfies \eqref{ap3} in Definition \ref{def_approx}. 

\textit{2) Partially Observed Systems:} \textcolor{black}{For a partially observed system, recall from Subsection \ref{subsection:info_examples} that an information state is given by the conditional range $\Pi_t = [[X_t|m_t]] \in \mathcal{B}(\mathcal{X}_t)$, where $\mathcal{B}(\mathcal{X}_t)$ is the set of all compact subsets of $\mathcal{X}_t$. We can construct an approximate information state by mapping each element $x_t \in \Pi_t$ to its closest element in the quantized space $\hat{\mathcal{X}}_t$. 
We denote this mapping by $\nu_t: \mathcal{B}(\mathcal{X}_t) \to 2^{\hat{\mathcal{X}}_t}$, where $2^{\hat{\mathcal{X}}_t}$ is the power set of $\hat{\mathcal{X}}_t$, and note that $\nu_t(\Pi_t) := \{\mu_t(x_t) \in \hat{\mathcal{X}}_t~|~x_t \in \Pi_t\}$ for any $\Pi_t \in \mathcal{B}(\mathcal{X}_t)$. Then, $\hat{\Pi}_t = \nu_t(\Pi_t)$ is an approximate information state for the partially observed system for all $t=0,\dots,T$} with $\epsilon_t = 2L_{d_t} \cdot \gamma_t$ and $\delta_t = 2 \gamma_{t+1} + 2 L_{\bar{f}_t} \cdot L_{h_{t+1}} \cdot L_{f_t} \cdot \gamma_t$, where $\gamma_{T+1} = 0$, and $L_{\bar{f}_t}$, $L_{h_{t+1}}$ and $L_{f_t}$ are Lipschitz constants of $\bar{f}_t$, $h_{t+1}$, and $f_t$, respectively (proof in Appendix D). 

\vspace{-12pt}

\color{black}

\subsection{Learning an approximate information state} \label{subsection:learning_algorithm}

In this subsection, we present an approach to learning an approximate information state from data.
We assume access to multiple trajectories $(Y_{t+1}, C_t, U_t : t = 0,\dots,T)$ generated using an exploratory control strategy. We denote the collection of such trajectories by the dataset $\mathcal{D}$. We begin by creating data-driven conditional ranges $\mathcal{K}_{t+1}^{\text{ob}} = [[Y_{t+1}~|~M_t, U_t]]^{\mathcal{D}}$ for all $t$ by combining the observations at time $t+1$ for trajectories in $\mathcal{D}$ that share common $(M_t, U_t)$ at time $t$. Similarly, we estimate the maximum incurred cost from the data at each $t$ as $C^{\max}_t = \max_{c_t \in [[C_t~|~M_t, U_t]]^{\mathcal{D}}} c_t$. Here note that the maximum exists because $[[C_t~|~M_t, U_t]]^{\mathcal{D}}$ is a finite set.

Then, we select three function approximators for each $t$: \textbf{(1) the encoder} $\psi_t: \hat{\mathcal{P}}_{t-1} \times \mathcal{U}_{t-1} \times \mathcal{Y}_{t} \to \hat{\mathcal{P}}_{t}$ that recursively compresses the history into an approximate information state, 
\textbf{(2) the observation decoder} $\phi_t: \hat{\mathcal{P}}_t \times \mathcal{U}_t \to \mathcal{B}(\mathcal{Y}_{t+1})$ that predicts next observations, where $\mathcal{B}(\mathcal{Y}_{t+1})$ is a set of compact subsets of $\mathcal{Y}$, 
and \textbf{3) the cost decoder} $\phi_t^{\text{c}}: \hat{\mathcal{P}}_t \times \mathcal{U}_t \to \mathcal{C}_t$ that predicts maximum cost.
The encoder is typically constructed with a recurrent neural network (e.g., GRU) whose hidden state is treated as $\hat{\Pi}_t = \psi_t(\hat{\Pi}_{t-1}, U_{t-1}, Y_t)$ at each $t$. The output of the encoder is augmented with the latest action $U_t$ to predict the conditional range $\hat{\mathcal{K}}_{t+1}^{\text{ob}} = [[Y_{t+1}~|~\hat{\Pi}_t, U_t]] = \phi_t(\hat{\Pi}_t, U_t)$, that best approximates $\mathcal{K}_{t+1}^{\text{ob}}$, and the worst-case cost $\hat{C}_t := \phi_t^{\text{c}}(\hat{\Pi}_t, U_t)$ that best approximates $C_t^{\text{max}}$. 

We also select two training losses $L^{\text{c}}_t = ||\hat{C}_t - C_t^{\max}||_2$ and $L^{\text{ob}}_t = \hat{\mathcal{H}}(\hat{\mathcal{K}}_{t+1}^{\text{ob}}, \mathcal{K}_{t+1}^{\text{ob}})$, where $\hat{\mathcal{H}}$ refers to a differentiable surrogate for the Hausdorff metric. For $\hat{\mathcal{H}}$, we can use either the surrogate procedure in \cite[Subsection II-B]{karimi2019reducing}, or the average Hausdorff distance \cite{aydin2021usage} between two sets $\mathcal{X},\mathcal{Y} \subset (\mathcal{S}, \eta)$:
\begin{gather*}
    \hat{\mathcal{H}}(\mathcal{X}, \mathcal{Y}) \hspace{-1pt} := \hspace{-1pt} \Big\{ \frac{ \int x \min_{y \in \mathcal{Y}} \eta(x,y) dx}{\int dx} \hspace{-1pt} + \hspace{-1pt} \frac{\int y \min_{x \in \mathcal{X}} \eta(x,y)}{\int dy}\Big\}.
\end{gather*}
Furthermore, we replace each $\min$ with a ``$\mathrm{softmin}$" to ensure differentiability of $\hat{\mathcal{H}}$. Then, the net training loss is defined at each $t$ as $L_t = \lambda \cdot L^{\text{c}}_t + L^{\text{ob}}_t$, where $\lambda \in \mathbb{R}_{\geq0}$ is a weight.

To learn an approximate information state model, we train the complete network assembly over each trajectory with loss $\sum_{t=0}^{T} L_t$ for one round of training. Further details on training are presented through a numerical example in Subsection \ref{subsection:pursuit}.
Given such a model learned from data, we can use a standard DP decomposition to derive an approximate control strategy.

\color{black}

\vspace{-6pt}

\section{Numerical Examples} \label{section:example}

We present two numerical examples to illustrate our approach: \textcolor{black}{\textit{(1) The Wall Defense Problem:} a worst-case control problem with partial observations, used to illustrate state-quantization from Subsection \ref{subsection:approx_examples}, and \textit{(2) The Pursuit Evasion Problem:} a worst-case reinforcement learning problem used to illustrate the learning approach in Subsection \ref{subsection:learning_algorithm}.}

\vspace{-12pt}

\subsection{The Wall Defense Problem}

In the wall defense problem, we consider an agent who defends a wall in a $5 \times 5$ grid world from an attacker over a time horizon $T$. The wall is located across the central row of the grid. We illustrate the wall defense problem for one initial condition in Fig. \ref{fig:1a}.
Here, the black-colored cells constitute the wall, and the grey-hatched cells are adjacent to the wall. The solid blue triangle, solid red circle and red ring are the agent, attacker and observation, respectively, at $t=0$. The pink cells are feasible positions of the attacker given the observation. The attacker moves within the bottom two rows of the grid and damages a wall cell when positioned in an adjacent cell. At each $t=0,\dots,T$, we denote the position of the attacker by $X_t^{\text{at}} \in \mathcal{X}^{\text{at}} = \{(-2,-1),$ $\dots,(2,-1),(-2, -2), \dots, (2,- 2) \}$. In contrast, the agent moves within the top two rows of the grid and repairs a wall cell when positioned in an adjacent cell. At each $t$, we denote the position of the agent by $X_t^\text{ag} \in \mathcal{X}^{\text{ag}} = \{(-2,1),\dots,(2,1),$ $ (-2,2), \dots, (2,2) \}$.
The state of the wall at each $t$ is the accumulated damage denoted by $D_t = (D_t^{-2}, \dots, D_t^2)$, where $D_t^i \in \mathcal{D}_t^i = \{0,1,2,3\}$ for all $i = -2, \dots, 2$ and $\mathcal{D}_t = \times_{i=-2}^2 \mathcal{D}_t^i$.
The attacker starts at the position $X_0^{\text{at}} \in \mathcal{X}^{\text{at}}$, which evolves for all $t$ as $X_{t+1}^{\text{at}} = \mathbb{I}(X_t^{\text{at}} + W_t \in \mathcal{X}^{\text{at}})\cdot(X_t^{\text{at}} + W_t) + (1$ $ - \mathbb{I}(X_t^{\text{at}} + W_t \in \mathcal{X}^{\text{at}}) )\cdot X_t^{\text{at}}$, where $\mathbb{I}$ is the indicator function and $W_t \in \mathcal{W}_t$ is an uncontrolled disturbance with $\mathcal{W}_t = \{(-1,0),$ $(1,0),(0,0),(0,1),(0,-1)\}$. 
At each $t$, the agent observes their own position and the wall's state. The agent also partially observes the attacker's position as $Y_t = \mathbb{I}(X_t^{\text{at}} + N_t \in \mathcal{X}^{\text{at}})\cdot(X_t^{\text{at}} + N_t) + (1-\mathbb{I}(X_t^{\text{at}} + N_t \in \mathcal{X}^{\text{at}}) )\cdot X_t^{\text{at}}$, where $N_t \in \mathcal{N}_t= \{(0,0), (0,1)\}$ is the measurement noise.
Given the history of observations, the agent selects an action $U_t \in \mathcal{U}_t = \mathcal{W}_t$ at each $t$. 
Starting with $X_0^{\text{ag}} \in \mathcal{X}^{\text{ag}}$, the agent moves as $X_{t+1}^{\text{ag}} = \mathbb{I}(X_t^{\text{ag}} + U_t \in \mathcal{X}^{\text{ag}})\cdot(X_t^{\text{ag}} + U_t) + (1 - \mathbb{I}(X_t^{\text{ag}} + U_t \in \mathcal{X}^{\text{ag}}) )\cdot X_t^{\text{ag}}$.
Starting with $D_0 = (0,0,0,0,0)$, the state of the wall evolves as $D_{t+1}^i = \min\big\{3, \max\big\{0, D_t^i + \mathbb{I}(X_t^{\text{at}} = (i,-1)) - \mathbb{I}(X_t^{\text{ag}} = (i, 1)) \big\}\big\}$ for all $t$ and $i = -2, \dots, 2$. At each $t$, after selecting the action, the agent incurs a cost for the damage to the wall, i.e., $c_t(D_t) = \sum_{i=-2}^2 D_t^i$. The agent's aim is to minimize the maximum instantaneous damage to the wall, i.e., $\mathcal{J}(\boldsymbol{g}) = \max_{t=0,\dots,T} \max_{x_0, w_{0:T}, n_{0:T}} c_t(D_t)$.

\begin{figure}[ht]
  \vspace{-10pt}
  \centering
  \subfigure[The original grid \label{fig:1a}]{\includegraphics[width=0.35\linewidth, keepaspectratio]{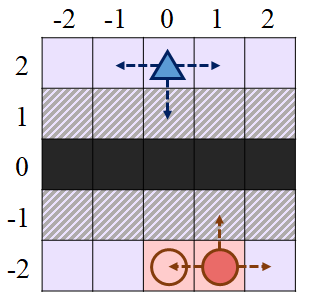}} 
  \subfigure[The quantized grid \label{fig:1b}]{\includegraphics[width=0.35\linewidth, keepaspectratio]{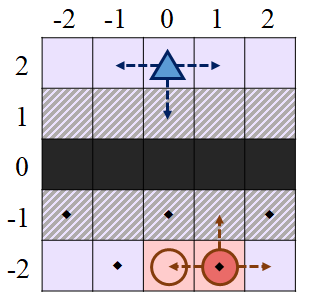}}
  \vspace{-10pt}
  \caption{The wall defense problem with the initial conditions $x_0^{\text{ag}} = (0,2)$ and $y_0 = (0,-2)$.}
  \label{fig:illustration_wall}
  \vspace{-10pt}
\end{figure}

\textcolor{black}{Since we are given a state-space model for this problem, this system is equivalent to a combination of Case 1 and Case 2 in Subsection \ref{subsection:info_examples}, i.e., the state of the agent and damage to the wall are perfectly observed, whereas the state of the attacker is partially observed. Thus, an information state at each $t$ is given by a combination of the two information states $\Pi_t = \big(X_t^{\text{ag}}, D_t, [[X_t^{\text{at}}|M_t]]\big)$. 
We construct an approximation of the conditional range $[[X_t^{\text{at}}|M_t]]$ at time $t$ using state quantization explained in Subsection \ref{subsection:approx_examples}, and define the approximate range $\hat{A}_t = \big\{ \mu_t(x_t) \in \hat{\mathcal{X}}^{\text{at}} | x_t \in [[X_t^{\text{at}}|M_t]] \big\}$. The set of quantized cells $\hat{\mathcal{X}}^{\text{at}}$, with $\gamma_t = 1$ for all $t$, is marked in Fig. \ref{fig:1b} with dots. Recall that $\mu_t(x_t) = \arg\min_{\hat{x}_t \in \hat{\mathcal{X}}}d(x_t,\hat{x}_t)$ and the approximate range at time $t$ is $\hat{A}_t  = \big\{ \mu_t(x_t) \in \hat{\mathcal{X}}^{\text{at}} | x_t \in [[X_t^{\text{at}}|M_t]] \big\}$. 
We consider the approximate information state $\hat{\Pi}_t = \big(X_t^{\text{ag}}, D_t, \hat{A}_t, Y_0\big)$ for all $t$.
Note that this approximate information state also includes the initial observation $Y_0$ in $\hat{\Pi}_t$ in addition to the quantized information state because it improves the prediction of $\hat{A}_{t+1}$ in practice. Including this additional term does not violate Definition \ref{def_approx}.} 
For five initial conditions, we compute the best control strategy for $T=6$ using both the information state (IS) and the approximate information state (AIS). In Fig. \ref{fig:results_defense}, we present the computational times (Run.) for both the DPs in seconds. 
Note that the approximate DP has a faster run-time in all cases. We also implement both strategies with random disturbances in the system with $T=6$. 
In Fig. \ref{fig:results_defense}, we also present the \textit{actual} worst-case costs across $5\times10^3$ implementations of both strategies and note that the AIS has a bounded deviation from the IS.

\begin{figure}[t]
  \centering  \includegraphics[width=\linewidth, keepaspectratio]{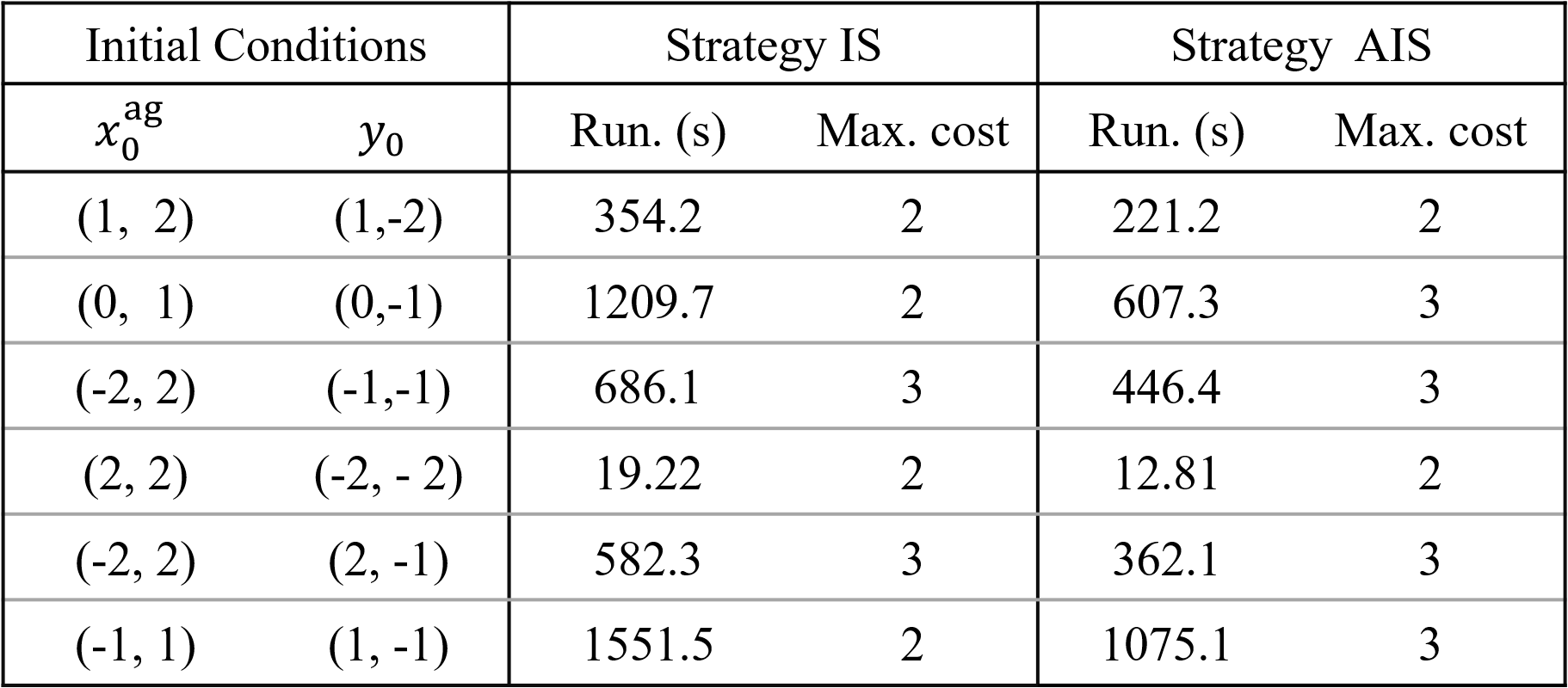} 
  \caption{Costs and run-times for $5\times10^3$ simulations and $T=6$.}
  \label{fig:results_defense}
  \vspace{-20pt}
\end{figure}

\vspace{-6pt}

\subsection{Pursuit Evasion Problem}
\label{subsection:pursuit}
In the pursuit-evasion problem, we consider an agent who chases a moving target in a $9 \times 9$ grid world with static obstacles. The agent aims to get close to the target over a time horizon $T$.
For each $t = 0,\dots,T$, we denote the position of the agent by $X_t^\text{ag} \in \mathcal{X}$ and that of the target by $X_t^{\text{ta}} \in \mathcal{X}$, where $\mathcal{X} = \big\{(-4,-4),\dots,(4,4)\big\} \setminus \mathcal{O}$ is the set of feasible grid cells and $\mathcal{O} \subset \mathcal{X}$ is the set of obstacles. 
The target starts at the position $X_0^{\text{ta}} \in \mathcal{X}$, which is updated as $X_{t+1}^{\text{ta}} = \mathbb{I}(X_t^{\text{ta}} + W_t$ $ \in \mathcal{X})\cdot(X_t^{\text{ta}} + W_t) + (1 - \mathbb{I}(X_t^{\text{ta}} + W_t \in \mathcal{X}) )\cdot X_t^{\text{ta}}$, where $W_t \in \mathcal{W}_t = \{(-1,0),(1,0),(0,0),(0,1),(0,-1)\}$ is the disturbance. 
At each $t$, the agent perfectly observes their own position and nosily observes the target's position as $Y_t = \mathbb{I}(X_t^{\text{ta}} + N_t \in \mathcal{X})\cdot(X_t^{\text{ta}} + N_t) + (1-\mathbb{I}(X_t^{\text{ta}} + N_t \in \mathcal{X}) )$ $\cdot X_t^{\text{ta}}$, where $N_t \in \mathcal{N}_t = \mathcal{W}_t$ is the measurement noise. Next, starting with $X_0^{\text{ag}} \in \mathcal{X}$, the agent selects an action $U_t \in \mathcal{U}_t = \mathcal{W}_t$ to move as $X_{t+1}^{\text{ag}} = \mathbb{I}(X_t^{\text{ag}} + U_t \in \mathcal{X})\cdot(X_t^{\text{ag}} + U_t) + (1 - \mathbb{I}(X_t^{\text{ag}} + U_t \in \mathcal{X}) )\cdot X_t^{\text{ag}}$.
At time $T$, the agent selects no action and observes the target's position $X_T^{\text{ta}}$ and incurs a cost $c_T(X_T^{\text{ta}}, X_T^{\text{ag}}) = \eta(X_T^{\text{ta}}, X_T^{\text{ag}}) \in \mathbb{R}_{\geq0}$, where $\eta$
is the shortest distance between two cells while avoiding obstacles. The distance between two adjacent cells is $1$ unit. 
The agent seeks to minimize the worst-case terminal cost \textit{without} prior knowledge of either the observation function or the target's evolution dynamics.
Note that this is a reinforcement learning generalization of Problem \ref{problem_1}.
We illustrate the grid and one initial setup in Fig. \ref{fig:pursuit_1a}.
Here, the black cells are obstacles. The solid blue triangle, solid red circle, and red ring are the agent, target, and observation, respectively, at $t=0$. The pink cells are feasible positions of the target given the observation.

\begin{figure}[t]
  \centering
  \subfigure[The original problem \label{fig:pursuit_1a}]{\includegraphics[width=0.32\linewidth, keepaspectratio]{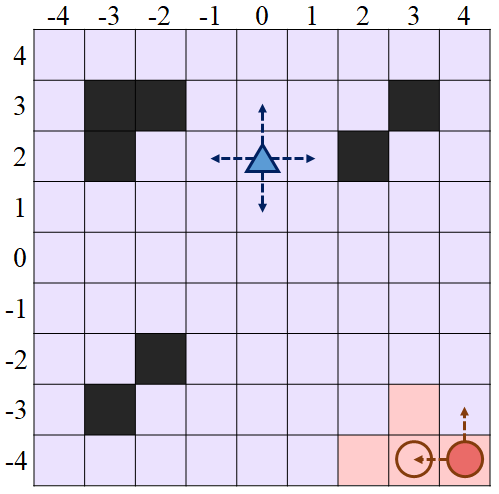}} 
  \subfigure[Actual observation prediction \label{fig:pursuit_1b}]{\includegraphics[width=0.32\linewidth, keepaspectratio]{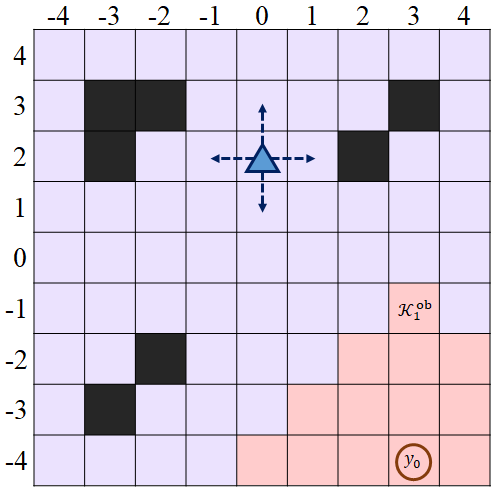}}
  \subfigure[Learned observation prediction \label{fig:pursuit_1c}]{\includegraphics[width=0.32\linewidth, keepaspectratio]{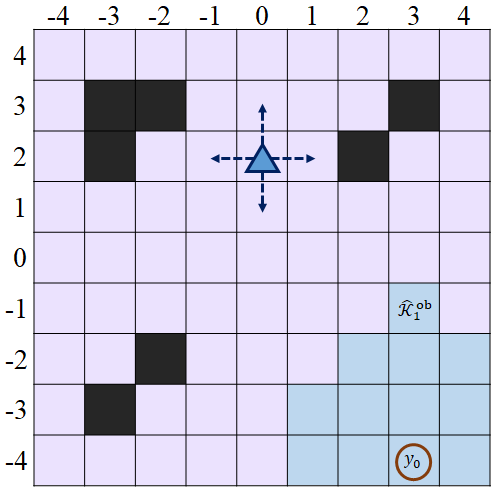}}
  \vspace{-8pt}
  \caption{The pursuit evasion problem with the initial conditions $x_0^{\text{ag}} = (0,2)$ and $y_0 = (3, -4)$.}
  \label{fig:illustration}
  \vspace{-20pt}
\end{figure}

\textcolor{black}{We approach this problem using the model-based reinforcement learning approach from Subsection \ref{subsection:learning_algorithm}, where we learn an approximate information state representation from $3 \times 10^7$ trajectories and then compute a control strategy.} 
We use the dataset $\mathcal{D}$ to construct estimates of the conditional range $\mathcal{K}^{\text{ob}}_{t+1}= [[Y_{t+1}|Y_{0:t}]]^{\mathcal{D}}$ for all $t = 0, \dots, T-2$ and $\mathcal{K}^{\text{ob}}_{T}= [[X_{T}^{\text{ta}}|Y_{0:T-1}]]^{\mathcal{D}}$. Then, we set up a deep neural network with an encoder-decoder structure for each $t=0,\dots,T$, as illustrated in Fig. \ref{fig:network}. 
At each $t$, the encoder $\psi_t$ comprises of 3 layer neural network with sizes $(2,14)$, $(14,12)$, $(12+24,24)$ and ReLU activation for the first two layers, where the inputs are  a $2$-d vector of coordinates for observation $Y_t$ and a $24$-d vector for the previous approximate information state $\hat{\Pi}_{t-1}$. The encoder compresses these inputs to a $24$-d vector representing the approximate information 
state $\hat{\Pi}_{t}$. 
At each $t$, the decoder $\phi_t$ is a 4-layer neural network of size $(24,48)$, $(48,56)$, $(56,64)$, $(64,74)$ with ReLU activation for the first three layers and sigmoid activation for the last layer. Its input is $\hat{\Pi}_t$, and its output is a $74$-d vector with each component taking values in $[0,1]$. Each component of the $74$-d output gives a set inclusion value for a specific feasible cell in the $9 \times 9$ grid, excluding obstacles. The output is thus interpreted as the conditional range $\hat{\mathcal{K}}^{\text{ob}}_{t+1} = [[Y_{t+1}~|~\hat{\Pi}_t]]$ for all $t=0,\dots,T-2$ and $\hat{\mathcal{K}}^{\text{ob}}_{T} = [[X_{T}^{\text{ta}}~|~\hat{\Pi}_{T-1}]]$.
We consider a set-inclusion threshold of $0.5$ for inclusion in $\hat{\mathcal{K}}^{\text{ob}}_{t+1}$ at each $t$.

\begin{figure}[ht]
  \vspace{-10pt}
  \centering  \includegraphics[width=\linewidth, keepaspectratio]{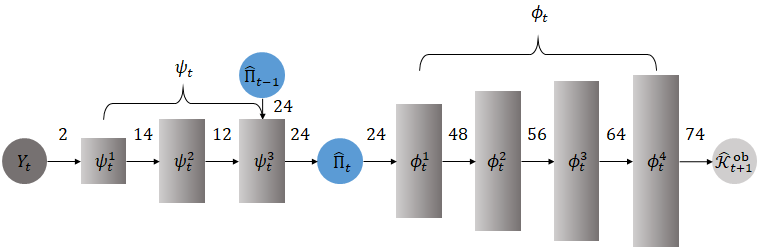} 
  \vspace{-10pt}
  \caption{The neural network architecture for approximate information states at any $t=0,\dots,T-1$.}
  \label{fig:network}
  \vspace{-8pt}
\end{figure}

The learning objective of our neural network at each $t$ is to minimize $\hat{\mathcal{H}}(\mathcal{K}^{\text{ob}}_{t+1},\hat{\mathcal{K}}^{\text{ob}}_{t+1})$, where we use the surrogate function in \cite[Subsection II-B]{karimi2019reducing}.
Note that at the terminal time step, this objective also minimizes the difference in maximum costs. 
We train the network for $40$ epochs using $90 \%$ of the available data with a learning rate of $0.0003$ and test it against the other $10 \%$. To illustrate the training results, consider an out-of-sample initial observation $y_0 = (3,-4)$. Then, the set $\mathcal{K}^{\text{ob}}_1$ constructed using data is shown by pink cells in \ref{fig:pursuit_1b} and the set $\hat{\mathcal{K}}^{\text{ob}}_1$ generated by of the trained network is shown by blue cells in \ref{fig:pursuit_1c}. Here, the trained network's output matches the conditional range constructed from data accurately except for one cell $(0,-4)$. We train a neural network for each $t$ up to $T=8$ to learn a complete approximate information state representation for the problem.
Then, at each $t$, the agent uses $(X_t^{\text{ag}}, \hat{\Pi}_t)$ in the approximate DP \eqref{DP_approx_1} - \eqref{DP_approx_2}.

We compare the performance of this approximate strategy with a baseline strategy that uses the observation $Y_t$ at each $t$ instead of $\hat{\Pi}_t$. Thus, for this baseline we train a network to match the prediction $[[Y_{t+1}~|~Y_t]]$ to $\mathcal{K}^{\text{ob}}_{t+1}$ for all $t=0,\dots,T-2$ and $[[X_T^{\text{ta}}~|~Y_{T-1}]]$ 
to $[[X_T^{\text{ta}}~|~Y_{0:T-1}]]$ at time $T-1$. The neural network structure is the same as before except for a lack of $\hat{\Pi}_{t-1}$ in the encoder input at each $t$, and we use the same training hyper-parameters as before. Subsequently, the agent computes an approximately optimal strategy using the approximate DP with the state $(X_t^{\text{ag}}, Y_t)$ at each $t$. 

\textcolor{black}{For six initial conditions, we present in Fig. \ref{fig:results_pursuit} the worst case costs obtained when implementing both the approximately optimal strategy (with AIS) and the baseline strategy (without AIS) for $T=4$ and $T=8$. Across $10^3$ simulations with randomly generated uncertainties, we note that using the learned approximate information state consistently improves worst-case performance compared to the baseline. Furthermore, the approximate strategy's outperformance grows larger for $T=8$ as compared to $T=4$ in most cases. Typically, we expect this margin to grow for longer time horizons because there is more loss of information when using only the latest observation for decision-making.
Thus, learning an approximate information state representation is a viable approach for robust model-based reinforcement learning.}

\begin{figure}[ht]
  \centering  \includegraphics[width=\linewidth, keepaspectratio]{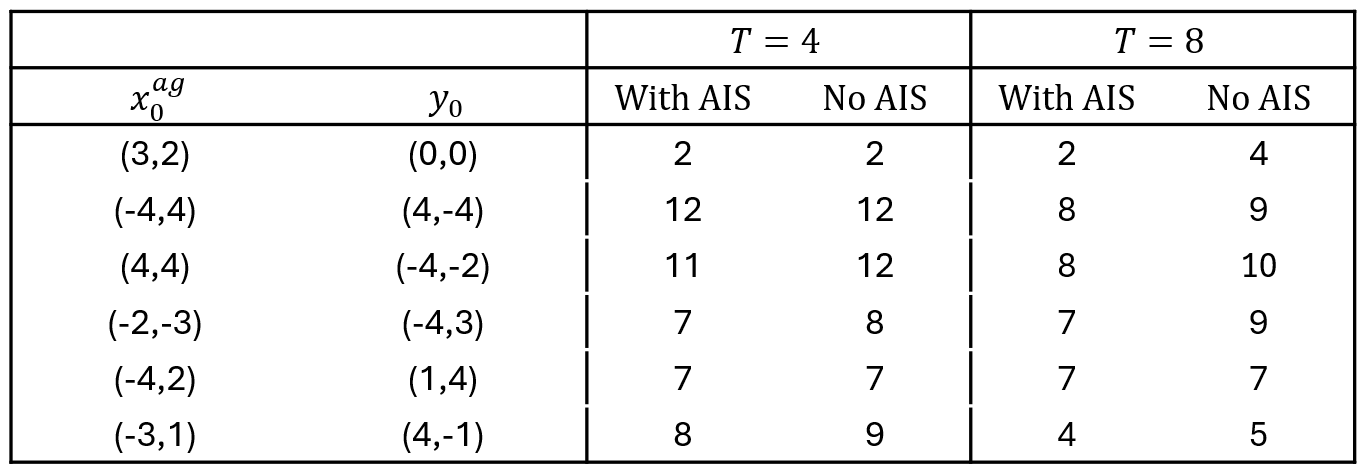} 
  \caption{Worst-case costs for $10^3$ simulations}
  \label{fig:results_pursuit}
  \vspace{-20pt}
\end{figure}

\section{Conclusion} \label{section:conclusion}

In this paper, we proposed a rigorous framework for worst-case control and learning in partially observed systems using non-stochastic approximate information states. We first presented two sets of properties to characterize information states and used them to construct a DP that yields an optimal control strategy. Then, we proposed two sets of properties to characterize approximate information states that can be constructed from output variables with knowledge of the dynamics or learned from output data with incomplete knowledge of the dynamics. We proved that approximate information states can be used in a DP to compute an approximate control strategy with a bounded loss in performance. We also presented theoretical examples of this bound and numerical examples to illustrate the performance of our approach in both worst-case control and reinforcement learning.

Our ongoing work is to specialize the approach in this paper to robust control \cite{dave2022additive} and reinforcement learning \cite{Dave2023infhorizon} in additive cost problems with partial observations. 
\textcolor{black}{While this paper presents a theoretical framework and preliminary numerical results for approximate information states, several key questions remain unanswered for future research. An explicit characterization of the trade-off between computational tractability and approximation quality is missing from the theory and would play a major role in practical application of approximate information states. It is also important to identify specific constructions of approximate information states beyond state quantization that improve performance in specific applications. Similarly, our reinforcement learning example is a toy problem with simple dynamics and a terminal cost criterion, presented with the goal of illustrating our theory. It is imperative that we study specific neural network architectures and learning algorithms to train efficient approximate models in realistic problems with large spaces and complex underlying interactions among various components.}

\vspace{-6pt}

\section*{Appendix A -- $L$-invertible Functions}

In this appendix, we present two classes of functions that are $L$-invertible: 1) all bi-Lipschitz functions that have a compact domain and a compact co-domain, and 2) all functions with a compact domain and a finite co-domain.

\begin{lemma}
Let $\mathcal{X}$ and $\mathcal{Y}$ be two compact subsets of a metric space $(\mathcal{S}, \eta)$. Then, any bi-Lipischitz function $f:\mathcal{X} \to \mathcal{Y}$ is $L$-invertible.
\end{lemma}

\begin{proof}
We begin by considering the pre-image set for any $y \in \mathcal{Y}$ under the function $f$. Note that the function $f$ is continuous because it is bi-Lipschitz, and the singleton $\{y\}$ is a compact subset of a metric space. Consequently, the pre-image $f^{-1}(y)$ is a bounded subset of $\mathcal{X}$.
Next, let $\mathcal{B}(\mathcal{X})$ denote the set of all bounded subsets of $\mathcal{X}$. Given the first result, we can consider a set-valued mapping $f^{-1}: \mathcal{Y} \to \mathcal{B}(\mathcal{X})$ which returns the pre-image for each $y \in \mathcal{Y}$. Then, for any $y^1, y^2 \in \mathcal{Y}$, using the definition of the Hausdorff distance in \eqref{H_met_def}:
\begin{multline}
    \mathcal{H}\big(f^{-1}(y^1), f^{-1}(y^2)\big)
    = \max\Big\{\sup_{x^1 \in f^{-1}(y^1)} \inf_{x^2 \in f^{-1}(y^2)} \\
    \eta(x^1, x^2),
    \sup_{x^2 \in f^{-1}(y^2)} \inf_{x^1 \in f^{-1}(y^1)}\eta(x^1, x^2)\Big\}. \label{L_inv_1_1}
\end{multline}
In the RHS of \eqref{L_inv_1_1}, the bi-Lipschitz property of $f$ implies that there exist constants $\underline{L}_{f}, \overline{L}_{f} \in \mathbb{R}_{>0}$ such that $\underline{L}_{f}\eta(x^1, x^2) \leq |f(x^1) - f(x^2)| \leq \overline{L}_{f}\eta(x^1, x^2),$
for all $x^1, x^2 \in \mathcal{X}$. Thus, for all $x^1 \in g^{-1}(y^1)$ and $x^2 \in g^{-1}(y^2)$, we write that
\begin{gather}
    \eta(x^1, x^2) \leq \underline{L}_{f}^{-1} \cdot \eta(y^1,y^2). \label{L_inv_1_2}
\end{gather}
The proof is complete by substituting \eqref{L_inv_1_2} into \eqref{L_inv_1_1} and defining the constant $L_{f^{-1}} := \underline{L}_f^{-1}$.
\end{proof}

\begin{lemma}
Let $\mathcal{X}$ be a compact subset and $\mathcal{Y}$ be a finite subset of $(\mathcal{S}, \eta)$. Then, any function $f:\mathcal{X} \to \mathcal{Y}$ is $L$-invertible.
\end{lemma}

\begin{proof}
Let $||\mathcal{Y}|| > 0 $ denote the minimum distance between two distinct elements in the finite, non-empty set $\mathcal{Y}$. Then, for any $y^1, y^2 \in \mathcal{Y}$ such that $y^1 \neq y^2$,
$\dfrac{\mathcal{H}\big(f^{-1}(y^1), f^{-1}(y^2)\big)}{\eta(y^1,y^2)}$
    $\leq \sup_{y^1, y^2 \in \mathcal{Y}}\dfrac{\mathcal{H}\big(f^{-1}(y^1), f^{-1}(y^2)\big)}{||\mathcal{Y}||} =: L_{f^{-1}},$
where $L_{f^{-1}} \in \mathbb{R}_{\geq0}$ is guaranteed to be finite because the set $\mathcal{X}$ is bounded and thus, so is the numerator. Thus, the function $f$ is $L$-invertible as defined in \eqref{L_inv_1_1}.
\end{proof}

\section*{Appendix B -- Preliminary Results}

In this subsection, we derive results necessary to prove the properties of the approximate DP in Subsection \ref{subsection:approx_properties}.

\begin{lemma} \label{lem_range_lipschitz}
Consider three bounded subsets $\mathcal{X}$, $\mathcal{Y}$ and $\mathcal{Z}$ of a metric space $(\mathcal{S},\eta)$. Let ${X} \in \mathcal{X}$, $Y \in \mathcal{Y}$ and $Z \in \mathcal{Z}$ be uncertain variables satisfying $Y = g(X)$, where $g: \mathcal{X} \to \mathcal{Y}$ is $L$-invertible, and $Z = h(X)$, where $h: \mathcal{X} \to \mathcal{Z}$ is Lipschitz. Then, there exists an $L_{Z|Y} \in \mathbb{R}_{\geq0}$ such that:
\begin{gather}
    \hspace{-4pt} \mathcal{H}([[Z|y^1]], \hspace{-1pt} [[Z|y^2]]) \hspace{-1pt} \leq \hspace{-1pt} L_{Z|Y} \hspace{-1pt} \cdot \hspace{-1pt} \eta(y^1,y^2), \; \forall  y^1, y^2 \in [[Y]]. \label{eq_range_lipschitz}
\end{gather}

\end{lemma}

\begin{proof}
We prove the result by constructing a feasible constant $L_{Z|Y} \in \mathbb{R}_{\geq0}$ which ensures that \eqref{eq_range_lipschitz} is satisfied for all $y^1, y^2 \in [[Y]]$. We begin by using the definition of the Hausdorff distance in \eqref{H_met_def} to expand the LHS of \eqref{eq_range_lipschitz} as
\begin{multline}
    \hspace{-6pt} \mathcal{H}\big([[Z|y^1]], [[Z|y^2]]\big) = \max\Big\{\sup_{x^1 \in g^{-1}(y^1)}\inf_{x^2 \in g^{-1}(y^2)} \hspace{-6pt} \eta\big(h(x^1), \\
    h(x^2)\big), \sup_{x^2 \in g^{-1}(y^2)}\inf_{x^1 \in g^{-1}(y^1)} \eta\big(h(x^1), h(x^2)\big)\Big\}, \label{eq_range_l_1}
\end{multline}
where, note that $[[Z|y]] = \big\{z \in \mathcal{Z}~|~ z = h(x), \forall x \in g^{-1}(y)\big\}$ for any realization $y \in [[Y]]$.
Next, recall that $h$ is Lipschitz continuous with a constant $L_{h} \in \mathbb{R}_{\geq0}$. 
Substituting this property into the RHS of \eqref{eq_range_l_1}, we write that
$\mathcal{H}\big([[Z|y^1]], [[Z|y^2]]\big)
    \leq L_h \cdot  \max\big\{ \sup_{x^1 \in g^{-1}(y^1)}$ $\inf_{x^2 \in g^{-1}(y^2)} \eta(x^1, x^2), \sup_{x^2 \in g^{-1}(y^2)}\inf_{x^1 \in g^{-1}(y^1)} \eta(x^1, x^2) \big\}$
    $= L_h \cdot \mathcal{H}\big(g^{-1}(y^1), g^{-1}(y^2)\big)$
    $= L_h \cdot L_{g^{-1}} \cdot \eta(y^1,y^2),$
where, in the second equality, we use the L-invertibile property of $g$. Then, the result follows by selecting $L_{Z|Y} := L_h \cdot L_{g^{-1}}$.
\end{proof}

\begin{lemma} \label{lem_prelim_2}
Consider a bounded set $\mathcal{X}$ and two functions $f:\mathcal{X} \to \mathbb{R}$ and $g:\mathcal{X} \to \mathbb{R}$. Then,
\begin{align}
    |\sup_{x \in \mathcal{X}} f(x) - \sup_{x \in \mathcal{X}} g(x)| &\leq \sup_{x \in \mathcal{X}}| f(x) - g(x)|, \label{prelim_2_1} \\
    |\inf_{x \in \mathcal{X}} f(x) - \inf_{x \in \mathcal{X}} g(x)| &\leq \sup_{x \in \mathcal{X}}| f(x) - g(x)|. \label{prelim_2_2}
\end{align}
\end{lemma}

\begin{proof}
We omit the proof due to space limitations.
Then, \eqref{prelim_2_2} follows from similar arguments as \eqref{prelim_2_1}.
\end{proof}

\begin{lemma} \label{lem_prelim_3}
For any four scalars $a,b,c,d \in \mathbb{R}$, 
\begin{align}
    |\max\{a,b\} - \max\{c,d\}| \leq \max\{|a-c|,|b-d|\}. \label{eq_prelim_3}
\end{align}
\end{lemma}

\begin{proof}
We omit the proof due to space limitations.
\end{proof}

\begin{lemma} \label{lem_prelim}
Consider two \textcolor{black}{nonempty} bounded subsets $\mathcal{A}, \mathcal{B}$ of a metric space $(\mathcal{X},\eta)$. Let $f: \mathcal{X} \to \mathbb{R}$ be a bounded continuous function with a Lipschitz constant $L_f \in \mathbb{R}_{\geq0}$ on $\mathcal{X}$. Then,
\begin{gather}
    \big| \sup_{a \in \mathcal{A}} f(a) - \sup_{b \in \mathcal{B}} f(b) \big| \leq L_f \cdot \mathcal{H}(\mathcal{A}, \mathcal{B}). \label{eq_prelim}
\end{gather}
\end{lemma}

\begin{proof}
We prove this result by considering three cases that are mutually exclusive but cover all the possibilities.
\textcolor{black}{Case 1: $\sup_{a \in \mathcal{A}} f(a) = \sup_{b \in \mathcal{B}} f(b)$, which implies that $|\sup_{a \in \mathcal{A}} f(a) - \sup_{b \in \mathcal{B}} f(b)| = 0$. The result holds directly from the fact that the RHS of \eqref{eq_prelim} is always non-negative.
Case 2: $\sup_{a \in \mathcal{A}} f(a) < \sup_{b \in \mathcal{B}} f(b)$,} which implies $| \sup_{a \in \mathcal{A}} f(a) - \sup_{b \in \mathcal{B}} f(b) | = \sup_{a \in \mathcal{A}} f(a) - \sup_{b \in \mathcal{B}} f(b) $. We define the \textcolor{black}{non-empty} set $\mathcal{A}^1(\beta) := \{a \in \mathcal{A} ~|~ f(a) + \beta \geq \sup_{b \in \mathcal{B}} f(b)\}$ for any infinitesimal $\beta > 0$. Then, $\sup_{a \in \mathcal{A}} f(a) - \sup_{b \in \mathcal{B}} f(b) \leq \sup_{a \in \mathcal{A}^1(\beta)} f(a) + \beta - \sup_{b \in \mathcal{B}} f(b) \leq \sup_{a \in \mathcal{A}^1(\beta)} \inf_{b \in \mathcal{B}}(f(a) - f(b)) + \beta \leq \sup_{a \in \mathcal{A}} \inf_{b \in \mathcal{B}}|f(a) - f(b)| + \beta \leq L_f \cdot \sup_{a \in \mathcal{A}} \inf_{b \in \mathcal{B}}\eta(a,b) + \beta$ for all $\beta > 0$. This implies that $|\sup_{a \in \mathcal{A}} f(a) - \sup_{b \in \mathcal{B}} f(b)| \leq L_f \cdot \sup_{a \in \mathcal{A}} \inf_{b \in \mathcal{B}}\eta(a,b) \leq L_f \cdot \mathcal{H}(\mathcal{A}, \mathcal{B})$, where, in the second inequality, we invoke the definition of the Hausdorff distance in \eqref{H_met_def} to complete the proof.
Case 3: $\sup_{a \in \mathcal{A}} f(a) < \sup_{b \in \mathcal{B}} f(b)$ and we can prove the result using the same sequence of arguments as case 1.
\end{proof}

Using Lemma \ref{lem_prelim}, we can also establish the following property. Consider two bounded sets $\mathcal{Y}, \mathcal{Z} \subset \mathbb{R}^n$, $n \in \mathbb{N}$. For two uncertain variables $Y \in \mathcal{Y}$ and $Z \in \mathcal{Z}$, let $[[Z|y]]$ satisfy $\mathcal{H}\big([[Z|y^1]], [[Z|y^2]]\big) \leq L_{Z|Y} \cdot \eta(y^1,y^2)$ for all realizations $y^1, y^2 \in \mathcal{Y}$ of $Y$. Then, for a continuous function $f:\mathcal{Z} \to \mathbb{R}_{\geq0}$, we use \eqref{eq_prelim} to state for all $y^1, y^2 \in [[Y]]$:
\begin{gather}
    \hspace{-8pt}\Big| \hspace{-1pt} \sup_{z^1 \in [[Z|y_1]]} \hspace{-4pt} f(z^1) - \hspace{-5pt} \sup_{z^2 \in [[Z|y_2]]} \hspace{-4pt} f(z^2) \Big| \hspace{-1pt} \leq L_{Z|Y} \hspace{-1pt} \cdot \hspace{-1pt} L_f \hspace{-1pt} \cdot \hspace{-1pt} \eta(y^1,y^2). \hspace{-2pt} \label{eq_prelim_5}
\end{gather}

\vspace{-6pt}

\section*{Appendix C -- Approximation Bounds for Perfectly Observed Systems}

In this appendix, we derive the values of $\epsilon_t$ and $\delta_t$ for all $t=0,\dots,T$ when an approximate information state is constructed using state quantization for a perfectly observed system, as described in Subsection \ref{subsection:approx_examples}. We first state a property of the Hausdorff distance, which we will use in our derivation.

\begin{lemma} \label{lem_union_prop}
Let $\mathcal{X}$ be a metric space with compact subsets
$\mathcal{A}, \mathcal{B}, \mathcal{C}, \mathcal{D} \subset \mathcal{X}$. Then, it holds that
\begin{align}
    \mathcal{H}\big(\mathcal{A} \cup \mathcal{B}, \mathcal{C} \cup \mathcal{D}\big)
    \leq \max \Big\{  \mathcal{H}\big(\mathcal{A}, \mathcal{C}\big) , \mathcal{H}\big(\mathcal{B}, \mathcal{D} \big) \Big\}. \label{eq_union_prop}
\end{align}
\end{lemma}

\begin{proof}
The proof is given in \cite[Theorem 1.12.15]{barnsley2006superfractals}.
\end{proof}

Next, we state and prove the main result of this appendix.

\begin{theorem} \label{thm_main_ap_a}
Consider a perfectly observed system, i.e., $Y_t = X_t$, for all $t=0,\dots,T$. Let $\mu_t: \mathcal{X}_t \to \hat{\mathcal{X}}_t$ such that $\max_{x_t \in \mathcal{X}_t} \eta (x_t, \mu_t(x_t)) \leq \gamma_t$ at each $t$. Then, $\hat{\Pi}_t = \mu_t(X_t)$ is an approximate information state which satisfies \eqref{ap1} with $\epsilon_t = 2L_{d_t} \cdot \gamma_t$ and \eqref{ap2} with $\delta_t = 2 \gamma_{t+1} + 2 L_{f_t} \cdot \gamma_t$ for all $t$, where $\gamma_{T+1} = 0$, and where $L_{d_t}$, and $L_{f_t}$ are Lipschitz constants for $d_t$ and $f_t$, respectively.
\end{theorem}

\begin{proof}
For all $t=0,\dots,T$, let $m_t = (x_{0:t}, u_{0:t-1})$ be the realization of $M_t$ and let the approximate information state be $\hat{x}_t = \mu_t(x_t)$. We first derive the value of $\epsilon_t$ in the RHS of \eqref{ap1}. At time $t$, can expand the conditional ranges to write that $[[X_t|m_t]] = [[X_t|x_t]] = \{x_t\}$ and $[[X_t|\hat{x}_t]] = \{x_t \in \mathcal{X}~|~ \eta (x_t, \hat{x}_t) \leq \gamma_t \}$. 
On substituting these into the LHS of \eqref{ap1}, we state that $\big|\sup_{c_t \in [[C_t|m_t, u_t]]} c_t - \sup_{c_t \in [[C_t|\mu_t(x_t), u_t]]} c_t\big|$
    $= \big|d_t(x_t,u_t) - \sup_{\bar{x}_t \in [[X_t|\mu_t(x_t)]]}d_t(\bar{x}_t,u_t)\big|$
    $\leq \sup_{\bar{x}_t \in [[X_t|\mu_t(x_t)]]} |d_t(x_t,u_t) - d_t(\bar{x}_t,u_t)|$
    $\leq L_{d_t} \cdot \sup_{\bar{x}_t \in [[X_t|\mu_t(x_t)]]} \eta(x_t, \bar{x}_t)$
    $\leq L_{d_t}\cdot\big(\eta(x_t, \mu_t(x_t)) + \sup_{\bar{x}_t \in [[X_t|\mu_t(x_t)]]} \eta(\mu_t(x_t), \bar{x}_t)\big),$
    $\leq 2 L_{d_t} \cdot \gamma_t =: \epsilon_t$, where, in the third inequality, we use the triangle inequality. Next, to derive the value of $\delta_t$, we expand the LHS of \eqref{ap2} as
\begin{align}
    &\mathcal{H}\big([[\hat{X}_{t+1}|x_t, u_t]],[[\hat{X}_{t+1}|\mu_t(x_t),u_t]]\big) \nonumber \\
    = &\mathcal{H}\big( \big\{ \mu_{t+1}(f_t(x_t,u_t,w_t))| w_t \in \mathcal{W}_t \big\},  \nonumber \\
    &\big\{ \mu_{t+1}(f_t(\bar{x}_t,u_t,w_t))| \bar{x}_t \in [[X_t|\mu_t(x_t)]], w_t \in \mathcal{W}_t \big\}\big) \nonumber \\
    \leq &\sup_{w_t \in \mathcal{W}_t} \mathcal{H}(\{ \mu_{t+1}(f_t(x_t,u_t,w_t))\},  \nonumber \\
    &\{ \mu_{t+1}(f_t(\bar{x}_t,u_t,w_t))| \bar{x}_t \in [[X_t|\mu_t(x_t)]]\}), \label{ap_a_1}
\end{align} 
where, we use \eqref{eq_union_prop} and $\big\{ \mu_{t+1}(f_t(x_t,u_t,w_t))| w_t \in \mathcal{W}_t \big\} = \cup_{w_t \in \mathcal{W}_t} \big\{\mu_{t+1}(f_t(\bar{x}_t,u_t,w_t))\big\}$. Using \eqref{eq_union_prop} in the RHS of \eqref{ap_a_1},
$\mathcal{H}\big([[\hat{X}_{t+1}|x_t, u_t]],[[\hat{X}_{t+1}|\mu_t(x_t),u_t]]\big) \hspace{-1pt}\leq \hspace{-1pt} \sup_{w_t \in \mathcal{W}_t, \bar{x}_t \in [[X_t|\mu_t(x_t)]]}\eta \big(\mu_{t+1}\big(f_t(x_t,u_t,w_t)\big),\mu_{t+1} \hspace{-2pt} \big(f_t(\bar{x}_t,$ 
$u_t,w_t)\big)\big)$ $\leq \sup_{w_t \in \mathcal{W}_t, \bar{x}_t \in [[X_t|\mu_t(x_t)]]} \big(\eta \big(\mu_{t+1}(f_t(x_t,u_t,w_t)),$ $f_t(x_t,u_t,
    w_t)\big) + \eta \big(f_t(x_t, u_t, w_t), f_t(\bar{x}_t,u_t,w_t)\big) +$ 
$\eta \big(f_t(\bar{x}_t,u_t,w_t), \mu_{t+1}(f_t(\bar{x}_t,u_t,w_t))\big)\big) \leq \gamma_{t+1} + 2L_{f_t} \gamma_t + \gamma_{t+1} =: \delta_t.$
\end{proof}

\vspace{-6pt}

\section*{Appendix D -- Approximation Bounds for Partially Observed Systems}

In this appendix, we derive the values of $\epsilon_t$ and $\delta_t$ for all $t=0,\dots,T$, when an approximate information state is constructed using state quantization for a partially observed system, as described in Subsection \ref{subsection:approx_examples}. 

\begin{theorem} \label{thm_main_ap_b}
Consider a partially observed system with $Y_t = h_t(X_t, N_t)$ for all $t=0,\dots,T$. Let $\mu_t: \mathcal{X}_t \to \hat{\mathcal{X}}_t$ such that $\sup_{x_t \in \mathcal{X}_t} \eta (x_t, \mu_t(x_t)) \leq \gamma_t$ at each $t$. Then, $\hat{\Pi}_t = \nu_t(\Pi_t)$ is an approximate information state with $\epsilon_t = 2L_{d_t} \cdot \gamma_t$ and $\delta_t = 2 \gamma_{t+1} + 2 L_{\bar{f}_t} \cdot L_{h_{t+1}} \cdot L_{f_t} \cdot \gamma_t$ for all $t$, where $\gamma_{T+1} = 0$, and where $L_{d_t}$, $L_{\bar{f}_t}$, $L_{h_{t+1}}$, and $L_{f_t}$ are Lipschitz constants for the respective functions in the subscripts.
\end{theorem}

\begin{proof}
For all $t=0,\dots,T$, let $m_t \in [[M_t]]$, $P_t = [[X_t|m_t]] \in \mathcal{P}_t$, and $\hat{P}_t = \nu_t(P_t) \in \hat{\mathcal{P}}_t$ be the realizations of the memory $M_t$, the conditional range $\Pi_t$ and the approximate information state $\hat{\Pi}_t$, respectively. Note that the conditional range $P_t$ satisfies \eqref{p1} and \eqref{p2} from Definition \ref{def_info_state}. Next, to derive the value of $\epsilon_t$, we write the LHS of \eqref{ap1} using \eqref{p1} as
\begin{align}
    &\big|\sup_{c_t \in [[C_t|m_t,u_t]]} c_t - \sup_{c_t \in [[C_t|\nu_t({P}_t),u_t]]}c_t \big| \nonumber \\
    =&\big|\sup_{x_t \in P_t} d_t(x_t, u_t) - \sup_{\bar{x}_t \in [[X_t|\nu_t({P}_t)]])}d_t(\bar{x}_t,u_t)\big| \nonumber \\
    \leq &L_{d_t} \cdot \mathcal{H}(P_t, [[X_t|\nu_t({P}_t)]]) \nonumber \\
    \leq &L_{d_t} \cdot \big(\mathcal{H}(P_t, \nu_t(P_t)) + \mathcal{H}(\nu_t(P_t), [[X_t|\nu_t({P}_t)]]) \big), \label{ap_b_1}
\end{align}
where, in the equality, we use \eqref{p1}; in the first inequality, we use \eqref{eq_prelim} from Lemma \ref{lem_prelim}; and in the second inequality, we use the triangle inequality for the Hausdorff distance. We can expand the first term in the RHS of \eqref{ap_b_1} as $\mathcal{H}(P_t, \nu_t(P_t)) = \mathcal{H}(P_t, \{\mu_t(x_t) \in \hat{\mathcal{X}}_t\;| \; x_t \in P_t\})$ $= \mathcal{H}\big(\cup_{x_t \in P_t} \{x_t\}, \cup_{x_t \in P_t}\{\mu_t(x_t) \in \hat{\mathcal{X}}_t\}\big) \leq \sup_{x_t \in P_t}$ $\eta(x_t, \mu_t(x_t)) \leq \gamma_t$, where we use \eqref{eq_union_prop} from Lemma \ref{lem_union_prop} in the first inequality. We can also expand the second term in the RHS of \eqref{ap_b_1} as $\mathcal{H}(\nu_t(P_t), [[X_t|\nu_t({P}_t)]]) ) = \mathcal{H}\big(\nu_t(P_t), \{x_t \in \mathcal{X}_t| \inf_{\bar{x}_t \in \nu_t(P_t)} \eta(x_t,\bar{x}_t) \leq \gamma_t \} \big) = \sup_{x_t \in [[X_t|\nu_t(P_t)]]} \inf_{\bar{x}_t \in \nu_t(P_t)} \eta(x_t, \bar{x}_t) \leq \gamma_t$, where the second equality holds by expanding the Hausdorff distance and noting that $\nu_t(P_t) \subseteq [[X_t|\nu_t(P_t)]]$. The proof is complete by substituting the results for both terms in the RHS of \eqref{ap_b_1}.

Next, to derive the value of $\delta_t$, we note that $P_t = \sigma_t(m_t)$. Then, using the triangle inequality in the LHS of \eqref{ap2},
\begin{align}
& \mathcal{H}\big([[\nu_{t+1}(\Pi_{t+1})|m_t,u_t]],[[\nu_{t+1}(\Pi_{t+1})|\nu_t(\sigma_t(m_t)),u_t]]\big) \nonumber \\
\leq & \mathcal{H}\big([[\nu_{t+1}(\Pi_{t+1})|m_t,u_t]],[[\Pi_{t+1}|m_t,u_t]]\big) + \nonumber \\
&\mathcal{H}\big([[\Pi_{t+1}|m_t,u_t]],[[\Pi_{t+1}|\nu_t(\sigma_t(m_t)),u_t]] \big) + \nonumber \\
& \mathcal{H}\big([[\Pi_{t+1}|\nu_t(\sigma_t(m_t)),u_t]],[[\nu_{t+1}(\Pi_{t+1})|\nu_t(\sigma_t(m_t)),u_t]]\big) \nonumber \\
\leq & 2\gamma_{t+1} + \mathcal{H}\big([[\Pi_{t+1}|m_t,u_t]],[[\Pi_{t+1}|\nu_t(\sigma_t(m_t)),u_t]]\big), \label{ap_b_2}
\end{align}
where, in the second inequality, we use the fact that $\mathcal{H}\big(P_{t+1}, \nu_{t+1}(P_{t+1})\big) \leq \gamma_{t+1}$, which was proved above. We can write the second term in the RHS of \eqref{ap_b_2} using \eqref{p2} from Definition \ref{def_info_state} as $\mathcal{H}\big([[\Pi_{t+1}|m_t,u_t]],[[\Pi_{t+1}|\nu_t(\sigma_t(m_t)),u_t]]\big) = \mathcal{H}\big([[\Pi_{t+1}|P_t,u_t]],[[\Pi_{t+1}|\nu_t(P_t),u_t]]\big)$. Furthermore, note that $[[\Pi_{t+1}|\nu_t(P_t), u_t]] = \big\{ \tilde{P}_{t+1} \in [[\Pi_{t+1}|\tilde{P}_t,u_t]]$ $| \tilde{P}_t \in [[\Pi_t|\nu_t(P_t)]] \big\} =  \cup_{\tilde{P}_t \in [[\Pi_t|\nu(P_t)]]} [[\Pi_{t+1}|\tilde{P}_t,u_t]]$. Next, we use \eqref{eq_union_prop} from Lemma \ref{lem_union_prop} to write that
\begin{align}
    & \mathcal{H}\big([[\Pi_{t+1}|P_t,u_t]]),[[\Pi_{t+1}|\nu_t(P_t),u_t]]\big) \nonumber \\
    \leq & \sup_{\tilde{P}_t \in [[\Pi_t|\nu_t(P_t)]]}\mathcal{H}\big([[\Pi_{t+1}|P_t,u_t]]),[[\Pi_{t+1}|\tilde{P}_t,u_t]]\big) \nonumber \\
    \leq & L_{\bar{f}_t} \cdot \sup_{\tilde{P}_t \in [[\Pi_t|\nu_t(P_t)]]}\mathcal{H}\big([[Y_{t+1}|P_t,u_t]]),[[Y_{t+1}|\tilde{P}_t,u_t]]\big) \nonumber \\
    \leq & L_{\bar{f}_t}  \hspace{-2pt} \cdot \hspace{-2pt} L_{h_{t+1}} \cdot  \hspace{-5pt} \sup_{\tilde{P}_t \in [[\Pi_t|\nu_t(P_t)]]} \hspace{-12pt} \mathcal{H}\big([[X_{t+1}|P_t,u_t]]),[[X_{t+1}|\tilde{P}_t,u_t]]\big), \label{ap_b_3}
\end{align}
where in the second inequality, we use the same arguments as in Lemma \ref{lem_alt_approx} and the third inequality can be proven by substituting $Y_{t+1} = h_{t+1}(X_{t+1},V_{t+1})$ into the equation. We can further expand the third term in the RHS of \eqref{ap_b_3} and use \eqref{eq_union_prop} from Lemma \eqref{lem_union_prop} to write that $\sup_{\tilde{P}_t \in [[\Pi_t|\nu_t(P_t)]]} \mathcal{H}\big([[X_{t+1}|P_t,u_t]]),[[X_{t+1}|\tilde{P}_t,u_t]]\big)$
    $\leq \sup_{\tilde{P}_t \in [[\Pi_t|\nu_t(P_t)]], w_t \in \mathcal{W}_t} \mathcal{H}\big( \{f_t(x_t,u_t,w_t)|x_t\in P_t\},$  $\{f_t(x_t,u_t,w_t)|x_t\in \tilde{P}_t\} \big)$
    $\leq L_{f_t} \cdot  \sup_{\tilde{P}_t \in [[\Pi_t|\nu_t(P_t)]]} \mathcal{H}\big(P_t, \tilde{P}_t\big)$
    $\leq L_{f_t} \cdot  \sup_{\tilde{P}_t \in [[\Pi_t|\nu_t(P_t)]]} \big(\mathcal{H}(P_t, \nu_t(P_t)\big) + \mathcal{H}\big(\nu_t(P_t), \tilde{P}_t)\big)$
    $\leq 2 L_{f_t} \cdot \gamma_t,$
where, in the third inequality, we use the triangle inequality and in the fourth inequality we use the fact that for all $\nu_t(\tilde{P}_t) = \nu_t(P_t)$, for all $\tilde{P}_t \in [[\Pi_t|\nu_t(P_t)]]$. 
\end{proof}

\bibliographystyle{ieeetr}

\bibliography{References,Latest_IDS}

\end{document}